\documentclass[aps,12pt,amsmath,amssymb,tightenlines,superscriptaddress]{revtex4}
\usepackage{amsmath,amssymb,amsthm}
\usepackage{graphicx}
\usepackage{booktabs}
\usepackage[usenames, dvipsnames]{color}
\usepackage{setspace}
\usepackage{rotating}
\usepackage{verbatim}

\usepackage[ plainpages = false, pdfpagelabels, 
                 pdfpagelayout = useoutlines,
                 bookmarks,
                 bookmarksopen = true,
                 bookmarksnumbered = true,
                 breaklinks = true,
                 linktocpage=all,
                 pagebackref=false,
                 colorlinks = true,
                 linkcolor = BrickRed,
                 urlcolor  = blue,
                 citecolor = BrickRed,
                 anchorcolor = green,
                 hyperindex = true,
                 hyperfigures
                 ]{hyperref} 

\newcommand{\cA}{\mathcal{A}}
\newcommand{\cB}{\mathcal{B}}
\newcommand{\cC}{\mathcal{C}}

\newcommand{\mui}{M}
\newcommand{\oui}{U}

\newcommand{\exA}{\textbf{A}}
\newcommand{\exB}{\textbf{B}}
\newcommand{\exC}{\textbf{C}}
\newcommand{\exD}{\textbf{D}}

\newcommand{\depspace}{\mathfrak{D}}
\newcommand{\diff}{\Delta}
\newcommand{\scorder}{s}
\newcommand{\significance}{\sigma}

\theoremstyle{plain}
\newtheorem{theorem}{Theorem}

\newtheorem{lemma}[theorem]{Lemma}
\newtheorem{corollary}[theorem]{Corollary}

\newtheorem*{thm1}{Theorem 1}

\theoremstyle{definition}
\newtheorem{defn}{Definition}

\begin{document}
\title{\href{http://necsi.edu/research/multiscale/}{An Information-Theoretic Formalism for Multiscale Structure in Complex Systems}}
\author{\href{http://www.people.fas.harvard.edu/~ballen/}{Benjamin Allen}}
\affiliation{Department of Mathematics, Emmanuel College}
\affiliation{Program for Evolutionary Dynamics, Harvard University}
\affiliation{Center for Mathematical Sciences and Applications, Harvard University}
\author{\href{http://www.sunclipse.org}{Blake C.\ Stacey}}
\affiliation{Martin A.\ Fisher School of Physics, Brandeis University}
\affiliation{New England Complex Systems Institute}
\author{\href{http://www.necsi.edu/faculty/bar-yam.html}{Yaneer Bar-\!Yam}}
\affiliation{New England Complex Systems Institute}
\date{\today}
\begin{abstract}
We develop a general formalism for representing and understanding
structure in complex systems. In our view, structure is the totality
of relationships among a system's components, and these relationships
can be quantified using information theory.  In the interest of
flexibility we allow information to be quantified using any function,
including Shannon entropy and Kolmogorov complexity, that satisfies
certain fundamental axioms.  Using these axioms, we formalize the
notion of a dependency among components, and show how a system's
structure is revealed in the amount of information assigned to each
dependency. We explore quantitative indices that summarize system
structure, providing a new formal basis for the complexity profile and
introducing a new index, the ``marginal utility of information''.
Using simple examples, we show how these indices capture intuitive
ideas about structure in a quantitative way. Our formalism also sheds
light on a longstanding mystery: that the mutual information of three
or more variables can be negative. We discuss applications to complex
networks, gene regulation, the kinetic theory of fluids and multiscale
cybernetic thermodynamics.
\end{abstract}

\maketitle

\section{Introduction}

\subsection{Overview}

The field of complex systems seeks to identify, understand and
predict common patterns of behavior across the physical, biological
and social sciences~\cite{baryam2003, haken2006, millerpage2007,
  boccara2010, newman2011}.  It succeeds by tracing these behavior
patterns to the structures of the systems in question.  We use the
term ``structure'' to mean the totality of quantifiable relationships
among the components comprising a system.  Systems from different
domains and contexts can share key structural properties, causing them
to behave in similar ways.  For example, the central limit theorem
tells us that we can sum over many independent random variables and
obtain an aggregate value whose probability distribution is
well-approximated as a Gaussian.  This helps us understand systems
composed of statistically independent components, whether those
components are molecules, microbes or human beings.  Likewise,
different chemical elements and compounds display essentially the same
behavior near their respective critical points.  The critical
exponents which encapsulate the thermodynamic properties of a
substance are the same for all substances in the same universality
class, and membership in a universality class depends upon structural
features such as dimensionality and symmetry properties, rather than
on details of chemical composition~\cite{sethna2006,
  ncatlab-universality}.

Mathematical representations of systems encode different aspects of
structure~\cite{baryam2003, millerpage2007, castellano2009,
  boccara2010}.  Common representations include
networks~\citep{watts1998, barabasi1999, strogatz2001,
  newman2003networks, wang2003, watts2003, boccaletti2006networks,
  dorogovtsev2010}, lattice models and cellular
automata~\citep{chopard1998cellular, hoekstra2010, stacey2011},
interacting agent models~\cite{liggett1985interacting, durrett1994,
  axelrod1997complexity, helbing2001traffic, bonabeau2002},
differential equations~\cite{mackey1977oscillation}, difference
equations~\cite{may1976simple,ott1981} and continuum field
theories~\cite{turing1952, cross1993pattern, maini1997spatial}.  While
mathematical inference is a rigorous aspect of science once a
mathematical model is identified, the choice of representation for a
particular system or class of systems often relies upon an \emph{ad
  hoc} leap of intuition.  The task of constructing useful
representations becomes challenging for complex systems, where the set
of system components is not just large, but also interwoven and
resistant to decomposition.  Indeed, the set of system components for
real-world physical, biological or social entities can be expected to
be so intricate that, for all practical purposes, precise enumeration
is ultimately intractable.  Achieving a fundamental solution to this
problem is critical for our ability to empower theoretical physics as
a general approach to complex systems, and a practical solution is
critical for our ability to address many real-world challenges.

An important clue about how to achieve a general solution is found in
the renormalization group analysis of phase transitions, a prototype
for complex-systems thinking. According to this analysis, we can
characterize analytically the set of ``relevant'' parameters that are
necessary and sufficient to characterize the behavior of a system in
the thermodynamic limit.  This provides a formal conceptual basis for
discarding variables which are unnecessary for a successful practical
description.  In addition, it furnishes a formal approach to obtaining
those variables, eliminating the \emph{ad hoc} aspects of constructing
models. Generalizing this approach requires an understanding of
information theory in a multiscale context, a context that has not
been developed in information theory nor in the statistical
physics of phase transitions. In particular we need a formal and
general understanding of how mathematical representations capture the
structure of systems, \emph{i.e.,} the information that characterizes
the set of possible configurations or behaviors of a system. We must
also consider the fundamental nature of the concept of scale in order
to distinguish, where appropriate, finer-scale details that can be
selectively neglected in favor of larger-scale ones.

These issues underline the need to develop a generalized mathematical
framework for discussing the multiscale structure of systems that
builds on information theory. Here we develop such a framework,
incorporating a set of complex-systems ideas, which can anchor
discussions of mathematical representation and formalize our intuitive
notion of system structure.  Our formalism applies to any system for
which a suitable quantitative measure of information can be defined.

The need for such a theory is also apparent from the lack of a
quantitative notion of a ``complex system'' as distinguished from
other varieties of systems.  It makes intuitive sense that
quantitatively defining system complexity should enable the
identification of complex systems.  However, efforts to define
complexity encounter a paradox. One might think to quantify system
complexity using measures like Kolmogorov/algorithmic complexity or
Shannon information that quantify irregularity or unpredictability in
an object or a stochastic process. However, the systems deemed the
most ``complex'' by these indices are those in which the components
behave independently of each other, such as ideal gases. Such systems
lack the multiscale regularities and interdependencies that
characterize the systems typically studied by complex systems
researchers.

Some theorists have argued that true complexity is best viewed as
occupying a position between order and
randomness~\cite{grassberger1986, crutchfield1989, crutchfield1994}.
Music is, so the argument goes, intermediate between still air and
white noise.  This answer is unsatisfying, however.  Though complex
systems contain both order and randomness, they do not appear to be
mere blends of the two.  For example, a box containing both a crystal
and an ideal gas is a system with intermediate entropy, but would not
normally be considered a complex system.  In a truly complex system,
the balance of order and randomness is captured in the multiscale
relationships among the system's components---that is, in the system's
structure.  Hence, we are brought to the realization that a formal
theory of structure based upon a generalization of information theory
is critical for understanding what complexity is and for
characterizing the essential attributes of complex systems.

We begin by outlining, in Section~\ref{Information}, the properties
which such an information function must satisfy.  Examples include,
but are not limited to, Shannon information, algorithmic complexity,
and vector space dimension.  We require only minimal assumptions so
that our formalism can apply as broadly as possible.

In Section~\ref{Systems}, we formalize the notion of a system and
introduce illustrative examples to which we will refer throughout the
work.  Section~\ref{Dependencies} introduces the central notion of a
dependency space---a Venn diagram or Euler diagram representation of
inter-component relationships.  Section~\ref{sec:Independence} sets
out the idea of a subsystem---one system embedded within another---and
Section~\ref{Scale} formalizes the idea of scale, elaborating how to
quantify multiscale relationships via information theory.

This development culminates in Section~\ref{Indices of Structure},
which discusses two indices of multiscale structure: the complexity
profile~\cite{baryam2004b} and a new index, the marginal utility of
information (MUI).  These indices resolve the paradox of complexity,
order and randomness, by showing how information and complexity can
exist at multiple scales.  The systems of greatest interest to the
complex systems community are those displaying nontrivial complexity
at a wide range of scales.  Section~\ref{Combinatorics of the
  Complexity Profile} develops a combinatorial formula for the
complexity profile and related quantities.

In Section~\ref{Special Classes of Systems} we consider special
classes of systems for which the indices introduced in
Section~\ref{Indices of Structure} take a simplified form.  We use
these systems to illustrate important properties of the two indices.
Section~\ref{Multiscale Cybernetic Thermodynamics} builds on these
ideas to study systems acted upon by external agents, using
multicylinder Szil\'ard engines as an illustrative example.  Finally,
Section~\ref{Discussion} presents our general conclusions and outlines
directions for future work.  This section discusses related concepts
including negentropy, requisite variety and their implications for the
scientific characterization of complex systems.

\subsection{List of key concepts}

\begin{itemize}
\item \emph{System:} A system is an entity composed of constituent
  parts, which we call \emph{components}.  Systems can be dynamic or
  static, deterministic or probabilistic.

\item \emph{Information:} Information quantifies the degree of
  freedom, irregularity or unpredictability of a set of components.
  Specific measures of information can be chosen depending on the type
  of representation available for the system, or for different
  purposes.  In each context, an information measure indicates how
  many questions one needs answered to remove uncertainty about
  the system components under consideration.  Though different
  measures of information are computed differently and require
  different types of data, they share certain fundamental mathematical
  properties, which we outline in Section \ref{Information}.

\item \emph{Dependency:} A dependency among a set of components is the
  relationship (if any) among them causing information pertinent to
  one component to be pertinent to them all.  We introduce a new
  notation for dependencies; for example, the dependency among
  components $a$, $b$, and $c$ is denoted $a;b;c$.  We also consider
  conditional dependencies such as~$a;b|c$, which stands for the
  relationship between $a$ and $b$ that exist independently of their
  relationships with~$c$.  The strength of the relationships that
  comprise a dependency can be quantified using information
  theory. For example, the conditional mutual information $I(a;b|c)$
  quantifies how strongly $a$ and $b$ are related in their behavior,
  excluding effects attributable to~$a$ and $b$'s mutual relationships
  with~$c$.
  
  \item \emph{Structure:} A system's structure is the totality of
  relationships among all sets of components, or, equivalently, the
  collection of all dependencies in the system.  We can characterize
  structure quantitatively in terms of the amounts of information
  assigned to each dependency.  Since this definition of structure
  makes no reference to the nature of the components or the mechanisms
  by which they interact, the structures of systems from very
  different contexts can be analyzed and compared using this
  framework.
  
\item \emph{Scale:} Relative size plays a central role in
  understanding and quantifying structure. The scale of a system
  behavior is given by the number of components that are engaged in
  that behavior. Formally, scale is the number of components involved
  in a dependency of the system. The extent of system behavior at a
  particular scale is quantified by the amount of information assigned
  to dependencies at this scale. Information and scale are
  complementary: As information is about the degree of freedom, scale
  is about constraints associated with redundancy. A large-scale
  behavior requires redundant information among the many components
  engaged in that behavior.

\item \emph{Indices of Structure:} Since systems with many components
  involve a large number of dependencies, the full structure of a
  system can be unwieldy to represent.  We develop two indices which
  give summary characterizations of a system's overall structure.  The
  first is the \emph{complexity profile}, an expression of the
  tradeoff between complexity and coordination introduced in prior
  work~\cite{baryam2004b}.  The second is a new measure, the
  \emph{marginal utility of information} (MUI).  These indices
  characterize, respectively, the amount of information that is
  present in the system behavior at different scales, and the
  descriptive utility of limited information through its ability to
  describe behavior of multiple components.

\end{itemize}

%%%%%%%%%%%%%%%%%%%%%%%%%%%%%%%%%%%%%%%%%%%%%%%%%%%%%%%%%%%%%%%%%%%%%%%%%%%%%%%%%%%%%%%%%%%%%%%

\section{Information}
\label{Information}

In defining the concept of structure we make use of a measure of
information. Conceptually, information specifies a particular entity
out of a set of possibilities and thus enables us to describe or
characterize that entity. A measure of information characterizes the
amount of information needed. Rather than adopting a specific
information measure, we consider that the amount of information may be
quantified in different ways, each appropriate to different
contexts. To unify these measures, we develop an axiomatically based
approach that considers a generalized \emph{information function}
satisfying two axioms. These axioms are satisfied by Shannon
information and algorithmic complexity among others. We use the
information function to map out how information is shared among
components of a system. This sharing---in which information about some
components can be gained by examining others---is central to our
discussion of structure.

Let $A$ be the set of components in a system.  An information
function, $H$, assigns a nonnegative real number to each subset $U
\subset A$, representing the amount of information needed to describe
the components in~$U$.  We require that such a function satisfy two
axioms:

\begin{itemize}
\item \emph{Monotonicity:} The information in a subset $U$ that is
  contained in a subset $V$ cannot have more information than $V$,
  that is, $H(U) \leq H(V)$.
\item \emph{Strong subadditivity:} Given two subsets, the information
  contained in both cannot exceed the information in each of them
  separately minus the information in their intersection:
\begin{equation}
H(U \cup V) \leq H(U) + H(V) - H(U \cap V).
\label{eq:strong-subadditivity}
\end{equation}
\end{itemize}

Strong subadditvity expresses how information combines when parts of a
system ($U$ and $V$) are regarded as a whole ($U \cup V$).
Information regarding $U$ may overlap with information regarding $V$
for two reasons.  First, $U$ and $V$ may share components; this is
corrected for by subtracting $H(U \cap V)$.  Second, constraints in
the behavior of non-shared components may reduce the information
needed to describe the whole. Thus, information describing the whole
may be reduced due to overlaps or redundancies in the information
applying to different parts, but it cannot be increased. These
redundancies are directly related to emergent collective behaviors.

The above axioms are best known in the context of Shannon entropy;
however, they apply to a number of measures that quantify information
or complexity, and different measures are appropriate for different
types of system:

\begin{itemize}

\item \emph{Microcanonical or Hartley entropy}: For a system with
a finite number of joint states, $H_0(U)=\log m$, where $m$
is the number of joint states available to the subset $U$ of
components.  Here, information content measures the number of
yes-or-no questions which must be answered to identify one joint state
out of~$m$ possibilities.

\item \emph{Boltzmann--Shannon entropy}: For a system characterized by
  a probability distribution over all possible joint states, $H(U) =
  -\sum_{i=1}^m p_i \log p_i$, where $p_1, \ldots, p_m$ are the
  probabilities of the joint states available to the components in
  $U$~\cite{shannon1948}.  Here, information content measures the
  number of yes-or-no questions which must be answered to identify one
  joint state out of all the joint states available to~$U$, where more
  probable states can be identified more concisely.

\item \emph{Algorithmic complexity}: For a system whose subsets can
  each be encoded as character strings, the algorithmic complexity
  $H(U)$ is the length of a maximally efficient description of~$U$
  according to some algorithmic scheme.  This notion has been
  formalized in a number of ways.  When a subset $U$ can be encoded as
  a binary string, the algorithmic complexity of~$U$ can be quantified
  as the length of the shortest self-delimiting program producing this
  string, with respect to some universal Turing machine.  Information
  content then measures the number of machine-language instructions
  which must be given to reconstruct $U$.  While conceptually clean,
  this definition is problematic.  First, the algorithmic complexity
  is only defined up to a constant which depends on the choice of
  universal Turing machine.  Second, thanks to the halting problem,
  the algorithmic complexity can only be estimated, not computed
  exactly.  We can establish upper bounds, but not precise values.
  These difficulties have led to modifications of the algorithmic
  complexity concept in which the description scheme is less
  wide-ranging than the set of all Turing machine
  programs~\cite{shallit2001, calude2009, ahnert2010}.
  
\item \emph{Logarithm of period}: For a deterministic dynamic system
  with periodic behavior, an information function can be defined as
  the logarithm of the period of a set of components (\emph{i.e.,} the
  time it takes for the joint state of these components to return to
  an initial joint state)~\cite{steudel2010}.  This information
  function measures the number of questions which one should expect to
  answer in order to locate the position of those components in their
  cycle.
  
\item \emph{Vector space dimension}: For a system the joint states of
  whose components can be described as points in a vector space, a
  possible information function is the dimension of the smallest vector
  space needed to describe the joint states of the components in~$U$.
  This dimension can be computed in practice, for example, by
  performing a principal components analysis on the variables
  representing components in~$U$~\cite{allen2008}.  Here, information
  content measures the number of coordinates one must specify in order
  to locate the joint state of~$U$.

\item \emph{Matroid rank}: A matroid consists of a set of elements
  called the \emph{ground set}, together with a \emph{rank function}
  that takes values on subsets of the ground set.  Rank functions are
  defined to include the monotonicity and strong subadditivity
  properties~\cite{dougherty2007}, and generalize the notion of vector
  subspace dimension. Consequently, the rank function of a matroid is
  an information function in our framework, with the ground set
  identified as the set of system components.
\end{itemize}

%%%%%%%%%%%%%%%%%%%%%%%%%%%%%%%%%%%%%%%%%%%%%%%%%%%%%%%%%%%%%%%%%%%%%%%%%%%%%%%%%%%%%%%%%%

\section{Systems}
\label{Systems}

\subsection{Definitions}
\label{sec:SystemsDefinitions}

We formally define a \emph{system} $\cA$ to be a finite set $A$ of
components, together with an information function $H_\cA$ (in this
case and for other definitions, we omit the subscript when only one
system is under consideration).  The choice of information function
will reflect how the system is modeled mathematically, and it affects
the kind of statements we can make about its structure.

A \emph{subsystem} is a smaller system embedded in a larger one.
Formally, we define a subsystem of~$\cA = (A,H_\cA)$ as a pair $\cB=(B,
H_\cB)$, where $B$ is a subset of~$A$ and $H_\cB$ is the restriction
of~$H_\cA$ to subsets of~$B$.

\subsection{Static, Dynamic, and Probabilistic Systems}

Systems can be static (existing in one state only), probabilistic
(existing in a number of possible states with associated
probabilities), or dynamic (existing in a sequence of states through
time).  Dynamic systems can be either deterministic or
stochastic.

Different information measures are appropriate depending on the
static, dynamic or probabilistic nature of the system in question.
For example, static systems may be amenable to algorithmic complexity
measures, whereas Shannon entropy applies naturally to probabilistic
systems.  Dynamic systems are most directly addressed as time
histories.  A single time history can be studied using algorithmic
measures, while an ensemble of time histories may be studied using
probabilistic measures.  Dynamic systems may also be treated as
probabilistic systems, using the approach of ergodic theory, wherein
the frequencies of occupancy of different states over extended periods
of time are treated as a probability distribution.  Our framework can
then characterize the structure of the system in terms of its
statistical behavior over long timescales.

The methods outlined here can be used to explore the dynamics of a
system's structure, using information measures whose values vary as
relationships change within a system over time.  However, our current
work focuses only on structure as an unchanging property of a system.

\subsection{Empirically motivated examples}

Our framework can be applied to a wide variety of real-world complex
systems.  We highlight four in particular:
\begin{itemize}
\item \emph{Gene regulatory systems:} Genes within a cell change over
  time in their expression levels, \emph{i.e.,} their rate of protein
  or RNA production.  Proteins produced by one gene may promote or
  inhibit the expression of other genes; thus, genes are an
  interdependent system with regard to their expression
  levels~\cite{jacob1961, britten1969gene, carey2001transcriptional,
    elowitz2002stochastic, lee2002, shen2002network, boyer2005,
    chowdhury2010}.  Individual genes can be represented as components
  of a system, and the information function quantifies the range of
  behaviors available to sets of genes.  Relationships, \emph{e.g.,}
  promotion or inhibition of one gene by another, can be quantified
  using mutual information.
\item \emph{Neural systems:} In a nervous system, neurons transmit
  electrical signals to each other.  These signals can be excitatory
  or inhibitory.  If the sum of input signals in a neuron exceeds some
  activation threshold, this neuron will ``fire'' and transmit signals
  to other neurons, promoting or inhibiting their firing in
  turn~\cite{hopfield1982neural, rabinovich2006}.  The components are
  neurons, and the information measure quantifies the range of joint
  spiking behavior in a collection of neurons (cf.~\cite{schneidman2006weak}).
\item \emph{Financial markets:} Financial markets are complex
  interdependent systems~\cite{mandelbrot1967distribution,
    mantegna1999, sornette2004stock, may2008, schweitzer2009,
    harmon2010, haldane2011systemic, harmon2011predicting, misra2011},
  where investors can be represented as system components, and the
  information function quantifies the range of investment activities
  among a set of investors.  Alternatively, one may view the assets as
  the components, and the information function quantifies the range of
  joint behavior in the prices of a set of assets.
\item\emph{Spin systems:} Many systems in statistical and
  condensed-matter physics are represented by considering components
  arranged on a graph or lattice.  The states of these components are
  characterized by discrete or continuous values, and these values
  vary stochastically according to which configurations are
  energetically favorable.  The contribution made by an individual
  component to the system's total energy depends on its interactions
  with its neighbors.  The prototypical example is the Ising model, in
  which each component has a ``spin'', which can be ``up'' or
  ``down'', and the interaction energy of a neighboring pair of spins
  depends on whether they are parallel or antiparallel.  Spin-system
  models play a vital role in the study of magnets, material mixtures
  such as alloys, liquid-gas phase transitions and other physical
  systems~\cite{domb1972, baryam2003, kardar2007b}.  The appropriate
  information function is the Shannon information, which is physically
  significant owing to the correspondence between Shannon information
  and thermodynamic entropy~\cite{feynman1996, sgs2004}.
\end{itemize}

\subsection{Simple examples}

To illustrate our formalism, we shall use four simple systems as examples.
\begin{itemize}
\item \emph{Example \exA: Three independent bits:} The
  system comprises three components, each of which is equally likely
  to be in state 0 or state 1, and the system as a whole is equally
  likely to be in any of its eight possible states.

\item \emph{Example \exB: Three completely interdependent bits:} Each
  of the three components is equally likely to be in state 0 or state
  1, but all three components are always in the same state.

\item \emph{Example \exC: Independent blocks of dependent bits:} Each
  component is equally likely to take the value 0 or 1; however,
  the first two components always take the same value, while the third
  can take either value independently of the coupled pair.

\item \emph{Example \exD: The $2+1$ parity bit system}: Three bits
  which can exist in the states 110, 101, 011, or 000 with equal
  probability.  Each of the three bits is equal to the parity (0 if
  even; 1 if odd) of the sum of the other two.  Any two of the bits
  are statistically independent of each other, but the three as a
  whole are constrained to have an even sum.
\end{itemize}

%%%%%%%%%%%%%%%%%%%%%%%%%%%%%%%%%%%%%%%%%%%%%%%%%%%%%%%%%%%%%%%%%%%%%%%%%%%%%%%%%%%%%%%%%%

\section{Dependencies}
\label{Dependencies}

Structure in complex systems reflects the observation that components
are not independent of each other.  This lack of independence implies
that the behavior or state of a component may then be inferred, in
whole or in part, from the behaviors or states of others.  We
illustrate this principle with three examples:
\begin{itemize} 
\item In the Example \exC~above, the state of the first component
  is determined by the state of the second, and vice versa.  In
  contrast, the state of the third component cannot be obtained from,
  nor used to obtain, the states of the first two.
\item In gene regulatory systems, the expression level of a gene may
  be determinable, in whole or in part, from the expression level of
  other genes that have regulatory interactions with this gene.
\item In a fixed structure such as a building, the components
  (\emph{e.g.,}~bricks, windows, etc.) are located in fixed spatial
  relationship to each other.  If the structure as a whole were moved
  in three-dimensional space, relative to some point of reference
  (which can be achieved by moving either the structure or the point
  of reference), the locations of three components would suffice to
  determine the locations of all others.
\end{itemize}

We call relationships such as these \emph{dependencies}.  Such
dependencies form the basis for our theory of structure.  With this
flexible notion of dependencies, our formalism describes not only
rigid structures such as a building, but also ``soft'' structures
arising from relationships that are not fully determinate,
\emph{e.g.,}~statistical or probabilistic relationships.  This section
introduces a general, information-theoretic language for describing
and quantifying dependencies.

\subsection{Notation for dependencies}

A \emph{dependency} among a collection of components $a_1, \ldots,
a_m$ is the relationship (if any) among these components such that the
behavior of some of the components is in part obtainable from the
behavior of others.  We denote this dependency by the expression
$a_1;\ldots;a_m$.  This expression represents a relationship, rather
than a number or quantity.  We use a semicolon to keep our notation
consistent with standard information theory (see below).

We can identify a more general concept of \emph{conditional
  dependencies}.  Consider two disjoint sets of components $a_1,
\ldots, a_m$ and $b_1, \ldots, b_k$.  The conditional dependency
$a_1;\ldots;a_m|b_1, \ldots, b_k$ represents the relationship (if any)
between $a_1,\ldots, a_m$ such that the behavior of some of these
components can yield improved inferences about the behavior of others,
relative to what could be inferred from the behavior of~$b_1, \ldots,
b_k$.  We call this the dependency of~$a_1, \ldots, a_m$ \emph{given}
$b_1, \ldots, b_k$, and we say $a_1, \ldots, a_m$ are \emph{included}
in this dependency, while $b_1, \ldots, b_k$ are \emph{excluded}.

A system's dependencies can be organized in a Venn diagram, as in
Figure \ref{fig:3comps}.  We call this diagram a \emph{dependency
  diagram.}  Each dependency in a system corresponds to a region of
the dependency diagram.

We call a dependency \emph{irreducible} if every system component is
either included or excluded. The irreducible dependencies in a
three-component system are pictured in Figure \ref{fig:3comps}. We
denote the set of all irreducible dependencies of a system $\cA$ by
$\depspace_\cA$.

The relationship between the components and dependencies of $\cA$ can
be captured by a mapping from~$A$ to subsets of~$\depspace_\cA$.  A
component $a \in A$ maps to the set of irreducible dependencies that
involve $a$ (or in visual terms, the region of the dependency diagram
that corresponds to component $a$).  We represent this mapping by the
function $\delta$.  For example, in a system of three components $a$,
$b$, $c$, we have
\begin{equation}
\delta(a) = \{ (a;b;c), \; (a;b|c), \; (a;c|b),  \; (a|b,c) \}.
\end{equation}
The parentheses around each dependency are used only to delineate
dependencies from each other.  We can extend the domain of this
function to subsets of components, by mapping each subset $U \subset
A$ onto to the set of all irreducible dependencies that involve at
least one element of~$U$; for example,
\begin{equation}
\delta(\{a,b\}) = \{ (a;b;c), \; (a;b|c), \; (a;c|b),  \; (b;c|a),  \; (a|b,c), \; (b|a,c)   \}.
\end{equation}
Visually, $\delta(\{a,b\})$ is the union of the circles representing
$a$ and $b$ in the dependency diagram. Finally, we can extend the
domain of this function to dependencies, by mapping the dependency
$a_1; \ldots; a_m|b_1, \ldots, b_k$ onto the set of all irreducible
dependencies that include $a_1, \ldots, a_m$ and exclude $b_1, \ldots,
b_k$; for example,
\begin{equation}
\delta(a|c) = \{ (a;b|c), \; (a|b,c) \}.
\end{equation}
Visually, $\delta(a|c)$ consists of the regions corresponding to~$a$
but not to~$c$.

\begin{figure}
\includegraphics[width=8cm]{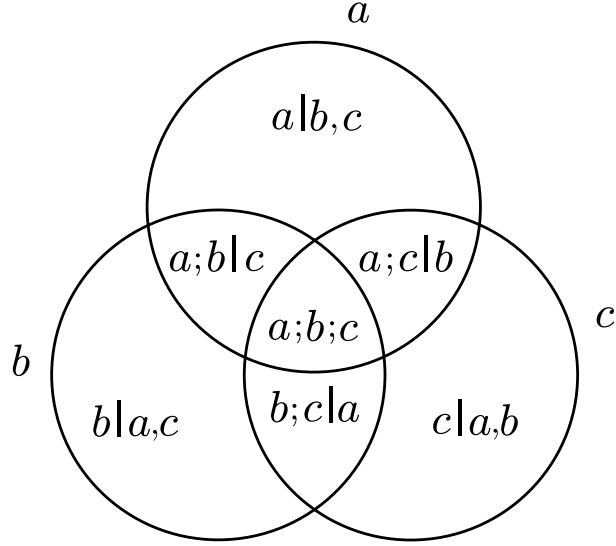}
\caption{The dependency diagram of a system with three components,
  $a$, $b$ and $c$, represented by the interiors of the three circles.
  The seven irreducible dependencies shown above correspond to the
  seven interior regions of the Venn diagram encompassed by the
  boundaries of the three circles.  Reducible dependencies such as
  $a|b$ are not shown.}
\label{fig:3comps}
\end{figure}

\begin{figure}[ht]
\includegraphics[width=8cm]{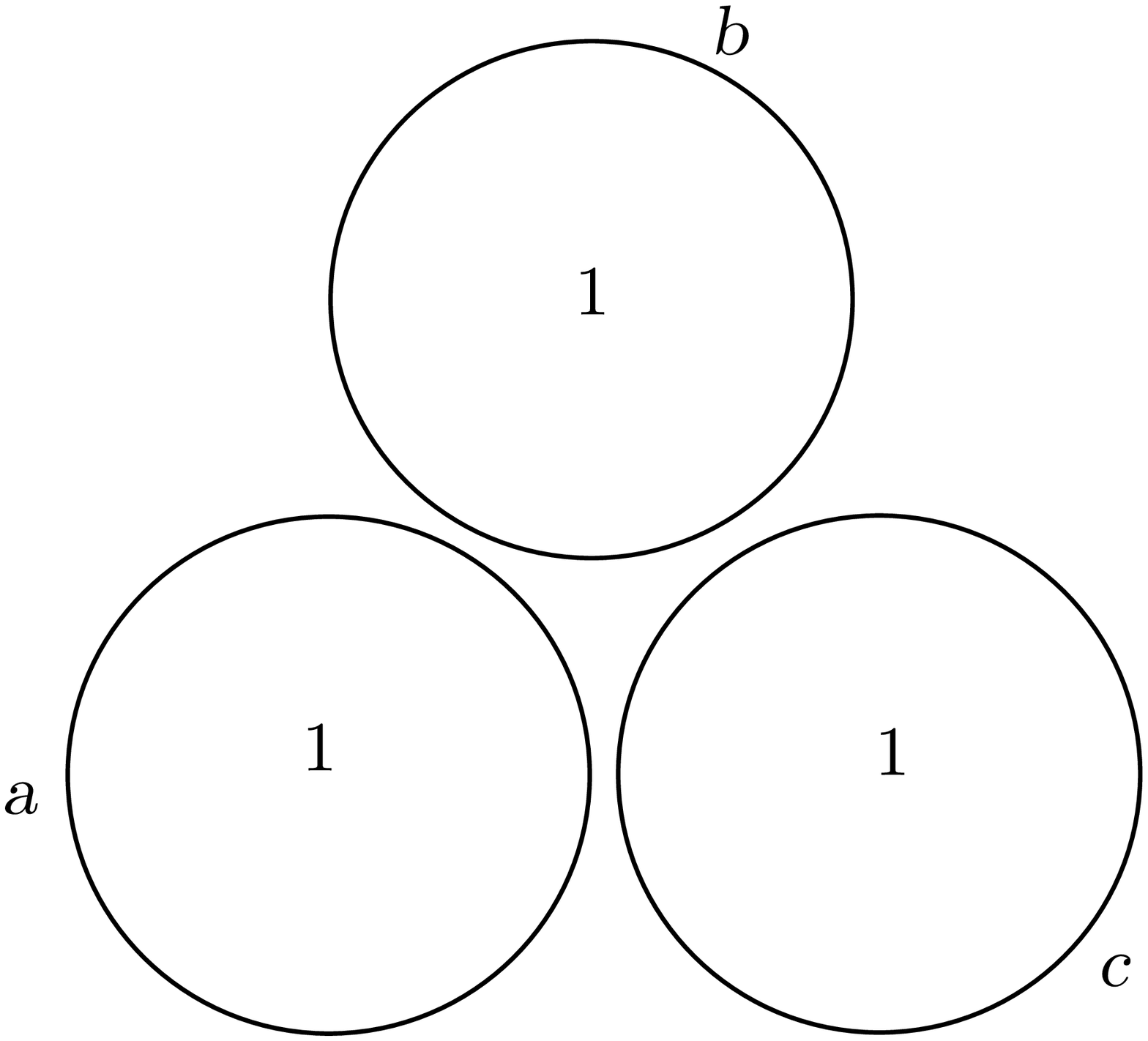}
\caption{\label{fig:exA} Dependency diagram for example \exA.}
\end{figure}

\begin{figure}[ht]
\includegraphics[width=5cm]{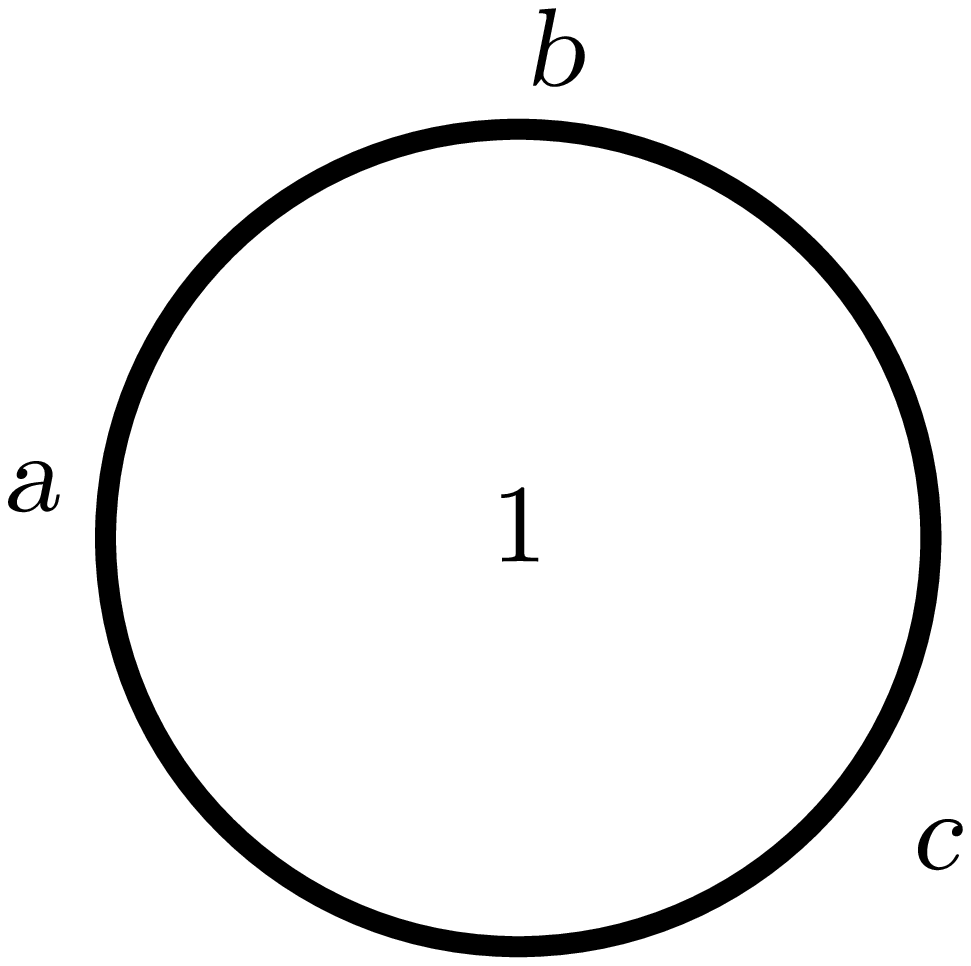}
\caption{\label{fig:exB} Dependency diagram for example \exB.}
\end{figure}

\begin{figure}[ht]
\includegraphics[width=6cm]{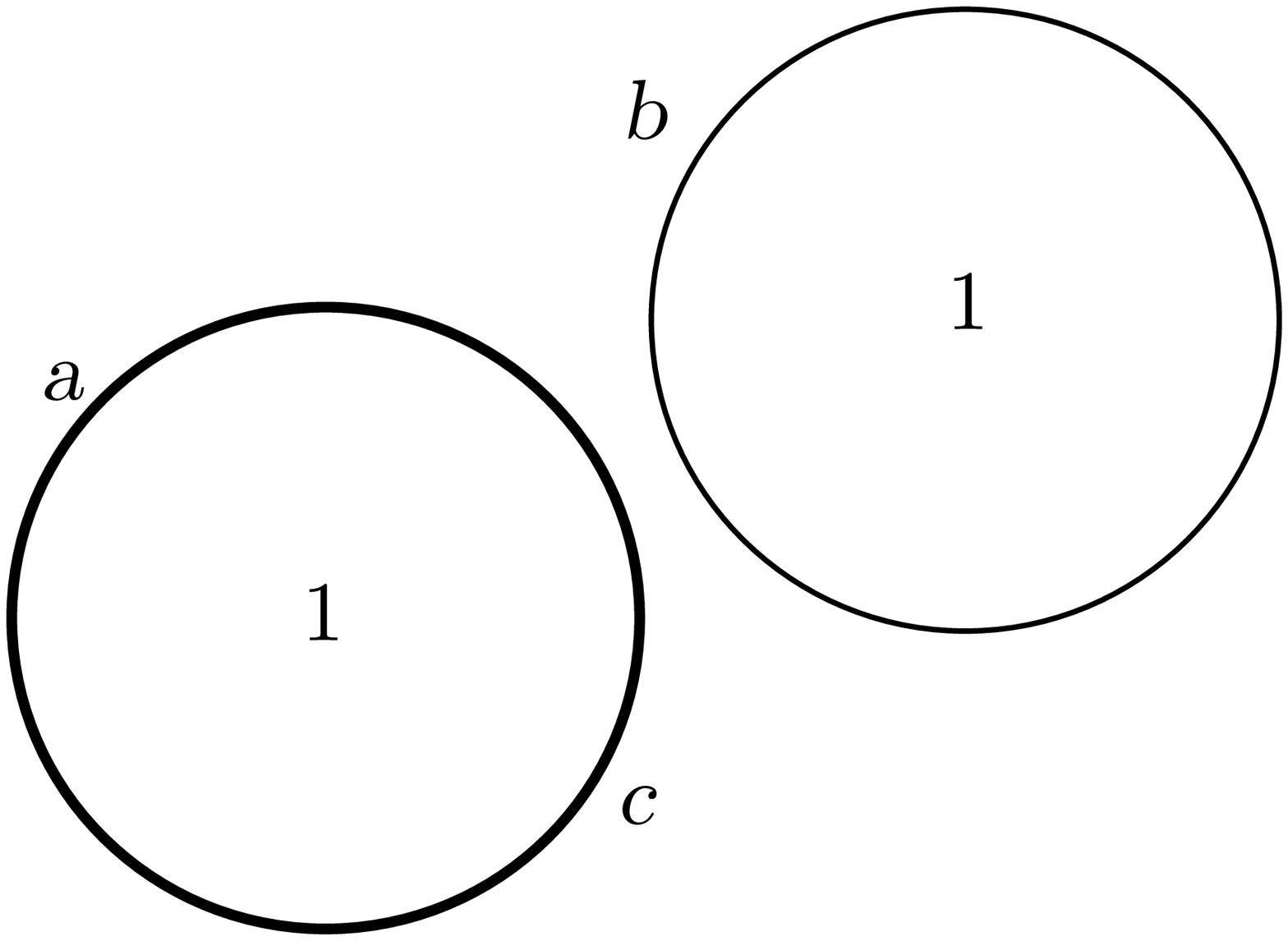}
\caption{\label{fig:exC} Dependency diagram for example \exC.}
\end{figure}

\begin{figure}[ht]
\includegraphics[width=8cm]{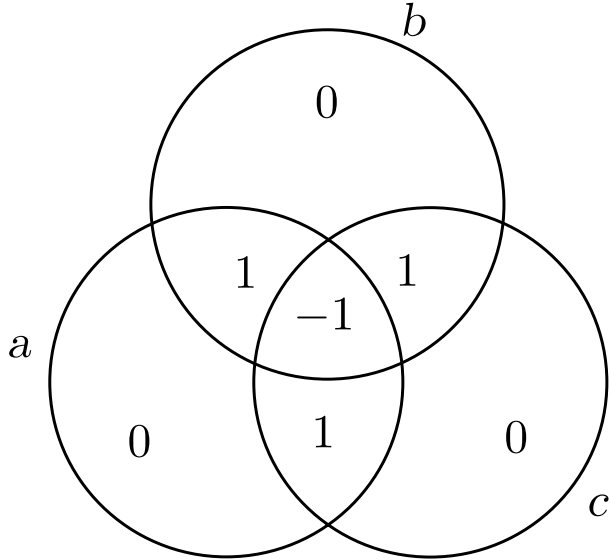}
\caption{\label{fig:paritybit} Dependency diagram for the $2+1$ parity
  bit system, example \exD.}
\end{figure}

\subsection{Information quantity in dependencies}
\label{InformationQuantity}

If a collection of components are dependent, such that the behavior of some
can be inferred from the behavior of others, this is reflected in
shared information among these components.  For example, in a system
of random variables, with Shannon entropy as an information function,
any statistical dependence between components $a$ and $b$ will cause
their joint information $H(a,b)$ to be less than the sum of their
separate informations $a$ and $b$, indicating the presence of shared
(mutual) information.  The amount of shared information quantifies the
strength of this dependence. Our formalism for system structure is
based on quantifying shared information across subsets of components
and multiple scales.

To formalize this idea, we introduce a function $I_\cA$ that
quantifies the shared information in the dependencies of a system
$\cA$.  The values of $I_\cA$ (the shared information in dependencies)
may be derived from the values of the previously defined information
function $H_\cA$, which is the joint information in sets of
components.  $H_\cA$ and $I_\cA$ characterize the same
quantity---information---but are applied to different kinds of
arguments: $H_\cA$ is applied to subsets of components of~$\cA$, while
$I_\cA$ is applied to dependencies.

To mathematically define the shared information $I_\cA$, we first
specify a solvable system of equations that determines its value on
irreducible dependencies.  For each subset $U \subset A$, we set
\begin{equation}
\label{eq:infosum}
\sum_{x \in \delta(U)} I(x) = H(U).  
\end{equation}
As $U$ runs over all subsets of $A$, the resulting system of equations
determines the values $I(x)$, $x \in \depspace_\cA$, in terms of the
values $H(U)$, $U \subset A$.  The solution is an instance of the
inclusion-exclusion principle \cite{erickson1996}, and can also be
obtained by Gaussian elimination.  An explicit formula obtained in the
context of Shannon entropy \cite{yeung1991} applies as well to any
information function.

We extend $I$ to dependencies that are not irreducible by defining the
shared information $I(x)$ to be equal to the sum of the values
of~$I(y)$ for all irreducible dependencies $y$ encompassed by a
dependency $x$:
\begin{equation}
\label{eq:Idependency}
I(x) = \sum_{y \in \delta(x)} I(y).
\end{equation}

More generally, we can extend the shared information $I$ to take, as
its argument, any set of irreducible dependencies $D
\subset \depspace_\cA$.  This is done by setting
\begin{equation}
I(D) = \sum_{y \in D} I(y).
\end{equation}
The above relation gives $\depspace_\cA$ the structure of a finite
signed measure space, with measure $I$. It is a signed measure space
because $I$ can take negative values (see below).

\begin{figure}[h]
\includegraphics[width=8cm]{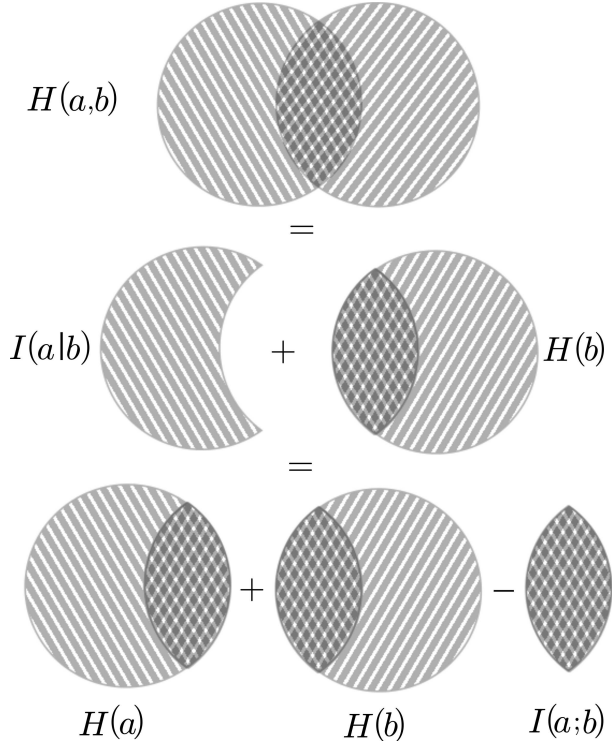}
\caption{\label{fig:venn-twocircles} An information diagram for a
  two-component system.  The information function $H$ is defined on
  sets of components, and the shared information $I$ is defined on
  dependencies.  In this picture, the circle representing component
  $a$ is shaded with left-leaning lines, and the circle for component
  $b$ is shaded with right-leaning lines.  The area which represents
  the mutual information is the area where both shadings overlap.  We
  can find the shared information of a dependency by adding and
  subtracting the values of $H$ for appropriate subsets.  For example,
  the mutual information $I(a;b)$ is $H(a) + H(b) - H(a,b)$, and the
  conditional information $I(a|b)$ is $H(a,b) - H(b)$.}
\end{figure}

The notation we use is chosen to correspond to that of Shannon
information theory.  This correspondence can be illustrated by
considering a system of random variables, with $H(U)$ representing the
joint Shannon entropy of variables in the set $U$.  In this case, $I$
represents the mutual and/or conditional information of a collection
of variables.  For instance, in a system with two random variables,
$a$ and $b$, solving \eqref{eq:infosum} yields
\begin{equation}
I(a;b) = H(a) + H(b) - H(a,b).
\end{equation}
This coincides with the classical definition of the mutual information
of $a$ and $b$ \cite{shannon1948, cover1991}.  It can similarly be
shown that, using Shannon information as the information function,
\begin{itemize}
\item $I(a_1|b_1, \ldots,b_k)$ is the conditional entropy of $a_1$
  given $b_1, \ldots,b_k$, and
\item $I(a_1; a_2|b_1, \ldots, b_k)$ is the conditional mutual
  information of $a_1$ and $a_2$ given $b_1, \ldots,b_k$.
\end{itemize}.

In general, for any information function $H$, we observe that the
information of one component conditioned on others, the shared
information $I(a_1|b_1, \ldots, b_k)$ is nonnegative due to the
monotonicity axiom.  Likewise, the mutual information of two
components conditioned on others, $I(a_1; a_2|b_1, \ldots, b_k)$, is
nonnegative due to the strong subadditivity axiom.

The crux of our formalism is that the collection of amounts of
information $I(x)$, as $x$ ranges over all dependencies of a system,
comprises a complete representation of a system's structure.
Formally, a system's structure is defined as the totality of
relationships among its components, and the collection of values
$I(x)$ provide a full quantitative characterization of those
relationships.  We have thus obtained our general representation of
structure; the remainder of this work will be concerned with
highlighting aspects of this representation that capture important
structural properties of systems.

While our formalism is built from the basic tools of information
theory---mutual and conditional information---the aims and scope of
our work depart from traditional information theory in a number of
directions.  First, information theory typically restricts its focus
to one or two variables at a time.  Multivariate mutual
information---the mutual information of three or more variables---has
been discussed in various contexts \cite{mcgill1980, han1980,
  yeung1991, caves1996, cerf1997, jakulin2003,
  baryam2004a,baryam2004b,baryam2004c, sgs2004, metzler2005,
  kolchinsky2011, james2011, baryam2012}, but it has not been
integrated into mainstream information theory, nor has its centrality
in the representation of complex systems been exploited.  Second,
information theory is primarily concerned with amounts of independent
bits of information; consequently, redundant information is typically
considered irrelevant, except insofar at it provides error correction
\cite{gallagher1968information,cover1991}. In contrast, we focus on
what redundant information reveals about relationships between
components and scales of behavior found in complex systems.  Third, by
defining information functions through their essential properties
(axioms) rather than by specific formulas, our formalism is applicable
to system representations for which traditional information measures
cannot be used.

We note that the study of multivariate information presents additional
challenges that do not arise in studying the information of only one
or two variables, due to the combinatorial number of quantities to be
calculated, the difficulty of calculation and in some cases the
difficulty of interpretation.  For instance, while the conditional
information $I(a_1|b_1, \ldots,b_k)$ and the conditional mutual
information $I(a_1; a_2|b_1, \ldots, b_k)$ are both nonnegative, the
mutual information of three or more variables can be negative.  Such
negative values appear to capture an important property of
dependencies, but the interpretation of these values as quantities of
information is somewhat counterintuitive.  As an example of negative
multivariate mutual information, consider the dependency diagram of
example \exD, as shown in Figure~\ref{fig:paritybit}.  The tertiary
shared information $I(a;b;c)$ is negative in this case.

%%%%%%%%%%%%%%%%%%%%%%%%%%%%%%%%%%%%%%%%%%%%%%%%%%%

\section{Independence}
\label{sec:Independence}

Independence is a central concept in the study of systems.  We define
independence by stating that components $a_1, \ldots, a_k$ of a system
$\cA=(A, H_\cA)$ are \emph{independent} of each other if their joint
information is equal to the sum of the information in each separately:
\begin{equation}
H(a_1, \dots, a_k) = H(a_1) + \ldots + H(a_k).
\end{equation}
This definition generalizes conventional notions of independence in
information theory, linear algebra and matroid theory.

We can extend this definition to apply at the level of
subsystems.  Subsystems $\cB_i=(B_i, H_{\cB_i})$ of~$\cA$, for~$i=1,
\ldots, k$, are defined to be independent of one another if
\begin{equation}
\label{eq:independence}
H_\cA(B_1 \cup \ldots \cup B_k) =  H_{\cB_1}(B_1) + \ldots + H_{\cB_k}(B_k).
\end{equation}
We recall from Section \ref{sec:SystemsDefinitions} that $H_{\cB_i}$ is the
  restriction of $H_\cA$ to subsets of $B_i$.

An immediate consequence of this definition is that if two subsystems are independent, they cannot have any components in common, except in the trivial case that each shared component has zero information.  A second important property, which we prove in Appendix \ref{app:independence}, is that if subsystems $\cB_1, \ldots, \cB_k$ are independent, then all components and subsystems of~$\cB_i$ are independent of all components and subsystems of~$\cB_j$ for all $j \neq i$.  In matroid theory, this is known as the hereditary property of independence \cite{perfect1981independence}.  For example, if subsystems $\{a\}$ and  $\{b,c\}$ of $\cA = \left(\{a,b,c\}, H_\cA \right)$ are independent, then $a$ and $b$ are independent and $a$ and $c$ are independent. The converse, however, is not true: In example \exD, $a$ is independent of $b$ and $a$ is also independent of $c$, but $\{a\}$ is not independent of $\{b,c\}$.   This occurs due to a global constraint among $a$, $b$ and $c$ that arises only when the three components are considered together (see Figure~\ref{fig:paritybit}).  More generally, for subsystems $\cB$ and $\cC$ of $\cA$, it is possible for $\cB$ to be independent of each subsystem of $\cC$ but not independent of $\cC$ itself.  This and other properties of independence are derived in Appendices \ref{app:independence} and \ref{app:MUIpbit} as consequences of the axioms of information.

%%%%%%%%%%%%%%%%%%%%%%%%%%%%%%%%%%%%%%%%%%%%%%%%%%%%%%%%%%%%%%%%%%%%%%%%%%%%%%%%%%%%%%%%

\section{Scale}
\label{Scale}

A defining feature of complex systems is that they exhibit nontrivial
behavior on multiple scales~\cite{baryam2003,baryam2004b,baryam2004c}.
For example, stock markets can exhibit small-scale behavior, as when
an individual investor sells a small number of shares for reasons
unrelated to overall market activity.  They can also exhibit
large-scale behavior, \emph{e.g.,}~a large institutional investor
sells many shares~\cite{misra2011}, or many individual investors sell
shares simultaneously in a market panic~\cite{harmon2011predicting}.

While the term ``scale'' has different meanings in different
scientific contexts, we use the term scale here in the sense of the
number of entities or units acting in concert.  We view scales as
additive, in that a collection of many individual components acting in
perfect coordination is regarded as equivalent to a single component,
whose scale is the sum of the scales of the individual components.

The notion of scale can be seen as complementary, or even orthogonal,
to the notion of information.  In the market panic example, since many
investors are doing the same thing, there is much overlapping or
redundant information in their actions---the behavior of one can be
largely inferred from the behavior of others.  Because of this
redundancy, the amount of information needed to describe their
collective behavior is low. This redundancy also makes this collective
behavior large-scale and highly significant.

\subsection{Scales of components}

For many systems, it is reasonable to regard all components as having
\emph{a priori} equal scale. In this case we may choose the units of
scale so that each component has scale equal to 1.  For other systems,
it is necessary to represent the components of a system as having
different intrinsic scales, reflecting their built-in size,
multiplicity or redundancy.  For example, in a system of many physical
bodies, it may be natural to identify the scale of each body as a
function of its mass, reflecting the fact that each body comprises
many molecules moving in concert.  In a system of investment banks
\cite{may2010systemic,haldane2011systemic,beale2011}, it may be
desirable to assign weight to each bank according to its volume of
assets.  In these cases, we denote the \emph{a priori} scale of a
system component $a \in A$ by~$\significance(a) >0$, defined in terms
of some meaningful scale unit.

\subsection{Scales of irreducible dependencies}
\label{ScaleOfDependency}

We can extend the notion of scale to apply to irreducible
dependencies.  Large-scale dependencies refer to relationships between
many components, and/or components of large intrinsic scale; whereas
small-scale dependencies refer to few components, and/or components of
small intrinsic scale.  The scale of a dependency may be considered to
quantify its importance to the system as a whole.

In a system with all components having equal scale, we define the
scale of an irreducible dependency $x \in \depspace_\cA$, denoted
$\scorder_\cA(x)$ or just $\scorder(x)$, to be the number of included
components.  This definition coincides with the intuitive
understanding of scale as the number of components acting in concert.
For example, in a system $\cA$ with $A=\{a,b,c\}$, the dependency
$a|b,c$ has scale 1, since it represents the behavior of $a$ that is
independent of $b$ and $c$.  The dependency $a;b;c$ has scale~3, since
it represents the behavior of $a$, $b$ and $c$ that is mutually
determinable.

If the components have different intrinsic scales, we define the scale
of a irreducible dependency $x$ to be
\begin{equation}
\label{eq:scaledef}
\scorder_\cA(x)= \sum_{\substack{a \in A \\ 
               x \text{ includes }a}} 
               \significance(a).
\end{equation}
In words, $\scorder_\cA(x)$ is the total scale of components included
in $x$, or, equivalently, the total number of scale units involved in
the mutually determinable behaviors represented by $x$.

\subsection{Scale-weighted information} \label{Significant Information}

A key concept in our analysis of structure---and a significant point
of departure from traditional information theory---is that in our
framework, any information about a system is understood as applying at
a specific scale.  This scale indicates the number of components, or
more generally, units of scale, to which this information
pertains. Information that is shared among a set of
components---arising from correlated or concerted behavior among these
components---has scale equal to the sum of the scales of these
components. In an insect swarm, for example, the motions of individual
insects are highly coordinated, so that there is a high degree of
overlap in information describing the motion of each insect; this
overlapping information therefore applies at a large scale.  In
emphasizing the scale at which information applies, we depart from
traditional information theory, which generally treats equal
quantities of information as interchangeable.

Since the scale of information quantifies the number of components or
units to which it applies, it is often natural to weight quantities of
information by their scale.  In this way, redundant information is
counted according to its multiplicity.  Scale-weighted information
helps characterize system structure, and plays a central role in the
quantitative indices of structure we explore in Section \ref{Indices
  of Structure}.

We define the \emph{scale-weighted information} $S(x)$ of an
irreducible dependency $x$ to be the scale of~$x$ times its
information quantity
\begin{equation}
S(x) = \scorder(x) I(x).
\end{equation}
We define the scale-weighted information of a subset $U
\subset \depspace_\cA$ of the dependence space to be the sum of the
scale-weighted information of each irreducible-dependency in this
subset:
\begin{equation}
\label{SIdef}
S(U) = \sum_{x \in U} S(x) = \sum_{x \in U} \scorder(x) I(x).
\end{equation}

The scale-weighted information of the entire dependency space
$\depspace_\cA$---that is, the scale-weighted information of the
system $\cA$---is invariant under changes in the system's structure.
Specifically, this total scale-weighted information is always equal to
the sum of the scale-weighted information of each component,
regardless of the relationships among these components. We state this
property in the following theorem, whose proof is given in
Appendix~\ref{app:SI}.
\begin{theorem}
\label{totalSI}
For any system $\cA$, the total scale-weighted information, $S(\depspace_\cA) = \sum_{x \in \depspace_\cA} s(x) I(x)$,  is given by the scale and information of each component, independent of the information shared among them:
\begin{equation}
S(\depspace_\cA) = \sum_{a \in A} \significance(a) H(a).
\end{equation}
\end{theorem} 

The total scale-weighted information, $S(\depspace_\cA)$, can thus be
considered a conserved quantity.  Its value does not change if the
system is reorganized or restructured.  This property arises directly
from the fact that scale-weighted information counts redundant
information according to its multiplicity; thus,
changes in information overlaps do not change the total.

%%%%%%%%%%%%%%%%%%%%%%%%%%%%%%%%%%%%%%%%%%%%%%%%%%%%%%%%%%%%%%%%%%%%%%%%%%%%%%%%%%%%%%%%%%%%%%%%%

\section{Quantitative indices of structure} \label{Indices of Structure}

Our definition of system structure as the amounts of information in
each of a system's dependencies presents practical difficulties for
implementation, in that the number of quantities grows combinatorially
with the number of system components.  It is therefore important to
have measures that summarize a system's structure.  Here we discuss
two such measures: the \emph{complexity profile} \cite{baryam2004b}
and a new measure, the \emph{marginal utility of information}.

\subsection{Complexity profile}
\label{sec:ComplexityProfile}

The complexity profile concretizes the observation that a complex
system is one which exhibits structure at multiple scales
\cite{baryam2003, baryam2012}.  The complexity profile of a system
$\cA$ is defined as a real-valued function $C_\cA(y)$ on the positive
real numbers whose value equals the total amount of information of
scale~$y$ or higher in~$\cA$:
\begin{equation}
\label{eq:Cdef}
C_\cA(y) = I \big(\{x \in \depspace_\cA\; : \; \scorder(x) \geq y\} \big).
\end{equation}
The complexity profile reveals the levels of interdependence in a
system. For systems where components are highly independent, $C(0)$ is
large and $C(y)$ decreases sharply in~$y$, since only small amounts of
information apply at large scales in such a system.  Conversely, in
rigid or strongly interdependent systems, $C(0)$ is small and the
decrease in~$C(y)$ is shallower, reflecting the prevalence of
large-scale information, as shown in in Figure 7.  We plot the
complexity profiles of our four running examples in Figure
\ref{fig:cp-toy-examples}.

\begin{figure}
\includegraphics[width=8cm]{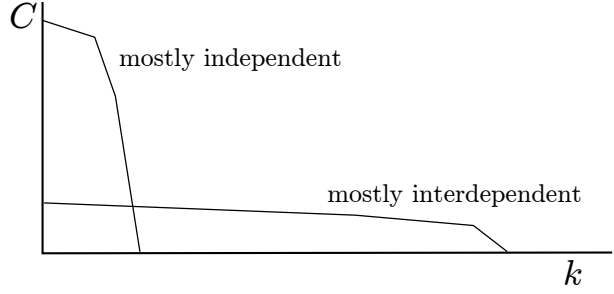}
\caption{\label{fig:CP} Complexity profiles for two systems, one whose
  components are largely independent of one another, and one whose
  components are strongly interdependent.}
\end{figure}

The complexity profile satisfies the following properties:
\begin{enumerate}
\item \emph{Conservation law:} The area under $C(y)$ is equal to the
  total scale-weighted information of the system, and is therefore
  independent of the way the components depend on each other
  \cite{baryam2004b}:
 \begin{equation}
 \int_0^\infty C(y) \; dy = S(\depspace_\cA).
\label{eq:cp-conservation}
 \end{equation} 
 This result follows from the conservation law for scale-weighted
 information, Theorem \ref{totalSI}, as shown in
 Appendix~\ref{app:conservation}.
\item \emph{Total system information:} At the lowest scale $y = 0$,
  $C(y)$ corresponds to the overall joint information: $C(0)=H(A)$.
  In particular, if Shannon entropy is the information function in
  question, then $C(0)$ is the Shannon entropy of the joint
  probability distribution for all the system's degrees of freedom.
  For physical systems, this is the total entropy of the system, in
  units of information rather than the usual thermodynamic units.
\item \emph{Largest scale of dependency:} If there are no interactions
  or correlations of scale $k$ or higher---formally, if $I(x)=0$ for
  all dependencies of scale greater than or equal to~$k$---then
  $C(y)=0$ for~$y \geq k$.  That is, the complexity profile vanishes
  for scales larger than the largest scale of organization within the
  system.
\item \emph{Additivity:} If a system $\cA$ is the union of two
  independent subsystems $\cB$ and $\cC$, the complexity profile of
  the full system is the sum of the profiles for the two subsystems,
  $C_\cA(y)=C_\cB(y) + C_\cC(y)$.  We prove this additivity property
  from the basic axioms of information functions in
  Appendix~\ref{app:additivity}.
\end{enumerate}

\begin{figure}[ht]
\includegraphics[width=8cm]{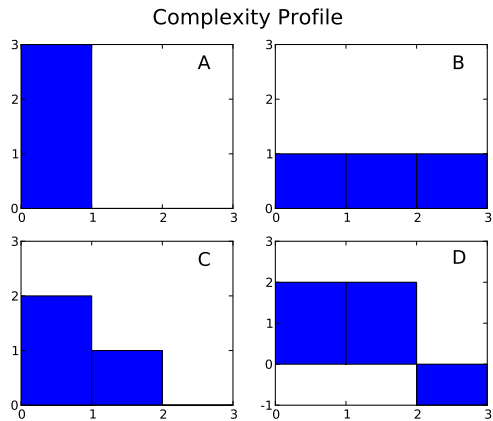}
\caption{\label{fig:cp-toy-examples} Complexity profile $C(k)$ for
  examples \exA~through \exD.  Note that the total (signed) area
  bounded by each curve equals $S(\depspace_\cA) = \sum_{a \in A}
  H(a)= 3.$}
\end{figure}

Due to the combinatorial number of dependencies for an arbitrary
system, calculation of the complexity profile may be computationally
prohibitive; however, computationally tractable approximations to the
complexity profile have been developed \cite{baryam2012}.

\subsection{Marginal utility of information}
\label{sec:MUI}
 
Here we introduce an alternative measure characterizing multiscale
structure: the \emph{marginal utility of information}, denoted
$\mui(y)$.  This index quantifies how well a system can be
characterized using a limited amount of information.

To obtain this measure, we first ask how much scale-weighted
information (as defined in Section \ref{Significant Information}) can
be represented using $y$ or fewer units of information.  We call this
quantity the \emph{maximal utility of information,} denoted $\oui(y)$,
and rigorously define it below.  For small values of~$y$, an optimal
characterization will convey only large-scale features of the system.
As $y$ increases, smaller-scale features will be progressively
included in the description.  For a given system $\cA$, the maximal
amount of scale-weighted information that can be represented,
$\oui(y)$, is constrained not only by the information limit $y$, but
also by the pattern of information overlaps in~$\cA$---that is, the
structure of~$\cA$.  More strongly interdependent systems allow for
larger amounts of scale-weighted information to be described using the
same amount of information $y$.
  
We define the marginal utility of information as the derivative of
maximal utility: $\mui(y) = \oui'(y)$.  $\mui(y)$ quantifies how much
scale-weighted information each additional unit of information can
impart.  The value of~$M(y)$, being the derivative of scale-weighted
information with respect to information, has units of scale.

The marginal utility of information has many properties similar to
those of the complexity profile, but with the axes reversed: the
argument of~$\mui(y)$ is information, while the value of~$\mui(y)$ has
units of scale.  Indeed, we show in Section \ref{Independent
  Subsystems} that, for a class of particularly simple systems, the
marginal utility of information and the complexity profile are
generalized inverses of each other.  $\mui(y)$ declines steeply for
rigid or strongly interdependent systems, and shallowly for weakly
interdependent systems.

We now develop the mathematical definition of the maximal utility
$U(y)$.  We call any entity $d$ that imparts information about system
$\cA$ a \emph{descriptor} of $\cA$.  The utility of a descriptor 
will be defined as a quantity of the form
\begin{equation}
\label{eq:utilitymaximand}
u=\sum_{a \in A} \significance(a) I(d; a).
\end{equation}
For this to be a meaningful expression, we consider each
descriptor $d$ to be an element of an augmented system $\cA^\dag=(A^\dag, H_{\cA^\dag})$,
whose components include $d$ as well as the original components of
$\cA$, which is a subsystem of $\cA^\dag$.  The amount of information that $d$ conveys
about any subset $V \subset A$ of components is given by
\begin{equation}
\begin{split}
I(d;V) & = I_{\cA^\dag}(d;V)\\  
& = H_{\cA^\dag}(d) + H_{\cA^\dag}(V)
 - H_{\cA^\dag} \big(\{d\} \cup V \big).
 \end{split}
\end{equation}
 For example, the amount that $d$ conveys about a component $a \in A$
 can be written $I(d;a)=H(d)+H(a)-H(d,a)$. $I(d;A)$ denotes the
 total information $d$ imparts about the system.  Because the original
 system $\cA$ is a subsystem of~$\cA^\dag$, the augmented information
 function $H_{\cA^\dag}$ coincides with~$H_\cA$ on subsets of~$A$.  
 
The quantities $I(d;V)$ are constrained by the structure of $\cA$ and the laws of information theory.  Applying the axioms of information functions to~$H_{\cA^\dag}$, we
 arrive at the following constraints on $I(d;V)$:
 \renewcommand{\labelenumi}{(\roman{enumi})}
\begin{enumerate}
\item $0 \leq I(d;V) \leq H(V)$ for all subsets $V \subset A$.
\item For any pair of nested subsets $W \subset V \subset A$, $ 0 \leq
  I(d;V)-I(d;W) \leq H(V)-H(W)$.
\item For any pair of subsets $V,W \subset A$, 
\begin{multline*}
I(d;V)+I(d;W) - I(d;V \cup W) - I(d;V \cap W)\\ \leq H(V) + H(W) - H(V
\cup W) - H(V \cap W).
\end{multline*}
\end{enumerate}
 
To obtain the maximum utility of information, we interpret the values $I(d;V)$ as variables subject to the above constraints. We define $U(y)$ as the maximum value of the utility expression, Eq.~\eqref{eq:utilitymaximand}, as $I(d;V)$ vary
subject to constraints~(i)--(iii) and that the total information $d$ imparts about $\cA$ is less than or
equal to~$y$: $I(d;A) \leq y$.  

$\oui(y)$ characterizes the maximal
amount of scale-weighted information that could in principle be conveyed about $\cA$ using $y$ or less units of information, taking into account the information-sharing in $\cA$ and the fundamental constraints on how information can be shared.  

$U(y)$ is well-defined since it is the maximal value of a linear
function on a bounded set.  Moreover, elementary results in linear
programming theory \cite{wets1966} imply that $\oui(y)$ is piecewise
linear, increasing and concave in~$y$.  It follows that $\mui(y)$ is
piecewise constant, positive and nonincreasing.

The marginal utility of information satisfies four properties
analogous to those satisfied by the complexity profile:
 \renewcommand{\labelenumi}{\arabic{enumi}.}
\begin{enumerate}
\item \emph{Conservation law:} The total area under the curve
  $\mui(y)$ equals the total scale-weighted information of the
  system:
\begin{equation}
\int_0^\infty \mui(y) \, dy = S(\depspace_\cA).
\label{eq:mui-conservation}
\end{equation}
This property follows from the observation that, since $\mui(y)$ is
the derivative of~$U(y)$, the area under this curve is equal to the
maximal utility of any descriptor, which is equal to~$S(\depspace_A)$
since utility is defined in terms of scale-weighted information.
\item \emph{Total system information:} The marginal utility vanishes for information values larger than the total system information, $\mui(y) = 0$ for $y > H(A)$, since, for higher values, the system has already been fully described. 
\item \emph{Largest scale of dependency:} If there are no interactions
  or correlations of degree $k$ or higher---formally, if $I(a_1;
  \ldots; a_k)=0$ for all collections $a_1, \ldots, a_k$ of~$k$
  distinct components---then $\mui(y)\leq k$ for all $y$.
\item \emph{Additivity:} If $\cA$ separates into independent
  subsystems $\cB$ and $\cC$, then
\begin{equation}
\label{eq:oui-additivity}
U_\cA(y) = \max_{\substack{y_1 + y_2 = y\\y_1,y_2 \geq 0}} \left(U_\cB(y_1) + U_\cC(y_2) \right).
\end{equation}
The proof follows from recognizing that, since information can apply either to $\cB$ or to $\cC$ but not both, an optimal description allots some amount $y_1$ of information to subsystem $\cB$, and the rest, $y_2 = y-y_1$, to subsystem $\cC$.  The optimum is achieved when the total maximal utility over these two subsystems is maximized.  Taking the derivative of both sides and invoking the concavity of $U$ yields a corresponding formula for the marginal utility $M$:
\begin{equation}
\mui_\cA(y) = \min_{\substack{y_1 + y_2 = y\\y_1,y_2 \geq 0}} \max
\big \{ \mui_\cB(y_1), \mui_\cC(y_2) \big \}.
\label{eq:mui-additivity}
\end{equation}  
Detailed proofs of Eqs.~\eqref{eq:oui-additivity} and \eqref{eq:mui-additivity} are provided in Appendix~\ref{app:mui}.   This additivity property can also be expressed as the reflection (generalized inverse) of $\mui$.  For any piecewise-constant, nonincreasing function $f$, we define the reflection $\tilde{f}$ as
\begin{equation}
\tilde{f}(x)  = \max\{y:f(y) \leq x \}.
\end{equation}
A generalized inverse \cite{baryam2012} is needed since, for piecewise constant functions, there exist $x$-values for which there is no $y$ such that $f(y)=x$.  For such values, $\tilde{f}(x)$ is the largest $y$ such that $f(y)$ does not exceed $x$.  This operation is a reflection about the line $f(y) = y$, and applying it twice recovers the original function. If $\cA$ comprises independent subsystems $\cB$ and $\cC$, the additivity property,  Eq.~\eqref{eq:mui-additivity}, can be written in terms of the reflection as 
\begin{equation}
\label{eq:mui-additivityreflection}
\tilde{\mui}_\cA(x) = \tilde{\mui}_\cB(x) + \tilde{\mui}_\cC(x).
\end{equation}
Eq.~\eqref{eq:mui-additivityreflection} is also proven in  Appendix~\ref{app:mui}.
\end{enumerate}

\begin{figure}[ht]
\includegraphics[width=8cm]{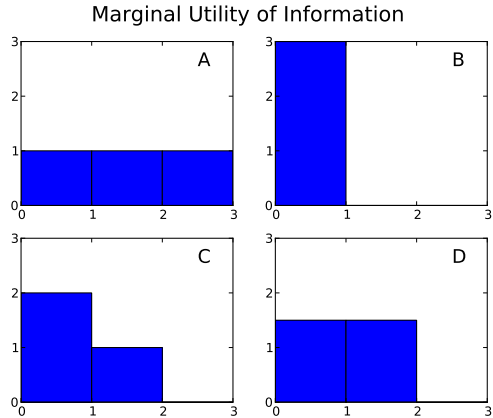}
\caption{\label{fig:mui-toy-examples} Marginal Utility of Information
  for examples \exA\ through \exD.  The total area under each curve is
  $\int_0^\infty \mui(y)\, dy = S(\depspace) = 3$.  Note that in
  examples \exA, \exB\ and \exC, the MUI is the reflection of the
  complexity profile shown in Figure~\ref{fig:cp-toy-examples}.  In
  examples \exA, \exB\ and \exC, the MUI curves can be thought of as
  sums of rectangles, one rectangle for each independent subsystem.}
\end{figure}

The MUI curves for our four running examples are shown in
Figure~\ref{fig:mui-toy-examples}.  Each curve is completely
determined by the dependency space of that system.  In each of the
four examples, the conservation law Eq.~(\ref{eq:mui-conservation})
implies that the total area under the MUI curve is 3.  We can deduce
from the ``largest scale of dependency'' property that for example
\exA, $\mui(y) \leq 1$ for all~$y$.  This suggests that the MUI curve
for example \exA{} should be a horizontal line at $\mui(y) = 1$ for $0
\leq y < 3$.  We can confirm this using the additivity property,
Eq.~(\ref{eq:mui-additivity}), because example \exA{} is a set of
three independent subsystems of one component each.  In
example \exB{}, a set of three fully correlated components, the
largest scale of dependency implies an upper bound on the MUI
of~$\mui(y) \leq 3$.  We deduce that the MUI curve for example \exB{}
should be $\mui(y) = 3$ for~$0 \leq y < 1$: describing one component
describes them all, so any descriptor having an information content of
1 or more can describe the whole system.  Example \exC{} can be broken
down into two independent subsystems, one of a single component and
the other consisting of a fully correlated pair.  Providing
information about the pair yields a higher return on investment, in
terms of scale-weighted information, than describing the isolated
component.  The MUI curve of example \exC{} is thus a horizontal line
$\mui(y) = 2$ for~$0 \leq y < 1$, which drops discontinuously
to~$\mui(y) = 1$ for~$1 \leq y < 2$, and falls to zero thereafter.

The most interesting case is the parity bit system, example \exD.
Symmetry considerations imply that a descriptor of maximal utility
conveys an equal amount of information about each of the three
components $a$, $b$ and $c$.  Constraints~(i)--(iv) then yield that
the amount described about each component must equal $y/2$ for~$0 \leq
y \leq 2$, and 1 for~$y>2$.  Thus the maximal utility is $U(y) = 3y/2$
for~$0 \leq y \leq 2$, and 3 for~$y>2$, and the marginal utility of
information is
\begin{equation}
M(y)=U'(y)= \begin{cases} \frac{3}{2} & 0 \leq y \leq 2\\
0 & y>2.
\end{cases}
\label{eq:MUI-exD-maintext}
\end{equation}
More generally, if an $N$-component system has a constraint which
manifests at the largest scale, and if the structure is symmetric as
it is in example~\exD, then
\begin{equation}
M(y)=U'(y)= \begin{cases} \frac{N}{N-1} & 0 \leq y \leq N-1\\
0 & y>N-1.
\end{cases}
\end{equation}
A detailed derivation is provided in Appendix \ref{app:MUIpbit}.  

The information overlaps
among the three components of example~\exD{} and the optimal
descriptor with information $y$ is illustrated in Figure~\ref{fig:mui-example-d}.  The
marginal utility of information $M(y)$  captures
the intermediate level of interdependency among components in the
parity bit system, in contrast to the maximal independence and maximal
interdependence in examples \exA~and \exB, respectively
(Figure~\ref{fig:mui-toy-examples}).

\begin{figure}[ht]
\includegraphics[width=8cm]{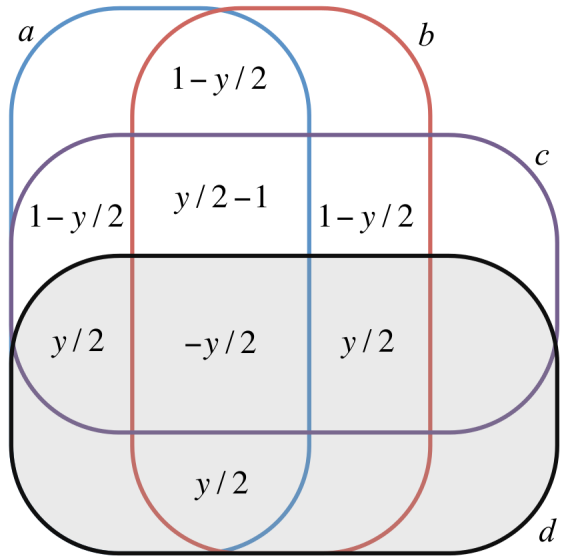}
\caption{\label{fig:mui-example-d} Information overlaps in the parity
  bit system, example \exD, augmented with a descriptor $d$ having
  information content $y \leq 2$ and maximal utility.  The amounts of
  information in the region corresponding to~$d$ sum to~$y$, the
  overall information of~$d$. The amounts of overlapping information
  among components $a$, $b$ and $c$, summing over the regions
  described and not described by~$d$, correspond to the amounts of
  information overlap in the original parity bit system, which are
  shown in Figure \ref{fig:paritybit}.}
\end{figure}

The idea of descriptors provides insight into negative values of
mutual information, as discussed in Section \ref{InformationQuantity}.
In example \exD, the tertiary shared information $I(a;b;c)$ is
negative.  Suppose there were a descriptor $d$ which applied only to
the irreducible dependency $a;b;c$ and not to any other irreducible
dependency.  That is, suppose $I_{\cA_\dag} (d;a;b;c)=-1$ and
$I_{\cA_\dag}(d;x)=0$ for any irreducible dependency $x$ of~$\cA$
other than $a;b;c$.  Then the total information in $d$, which equals
the sum of $d$'s shared information with all irreducible dependencies
of $\cA$, would be negative one: $h_{\cA_\dag}(d)=-1$.  This negative
information is impossible according to our axioms.  Thus the
(negative) amount of shared information associated with the
triple-overlap region cannot be described on its own.  It can,
however, be described implicitly as other aspects of the system are
described.  For instance, a \emph{complete} description of the parity
bit system, which contains full information about all three components
($I_{\cA_\dag}(d;a)=I_{\cA_\dag}(d;b)=I_{\cA_\dag}(d;c)=1$),
implicitly contains all the information assigned to the dependency
$a;b;c$.  The presence of information which can only be described
implicitly, rather than directly, has a physical meaning which we
explore in Section~\ref{Multiscale Cybernetic Thermodynamics}.

The MUI is closely connected to a number of other important quantities
studied in different fields of science, a point we will examine in the
Discussion section.

%%%%%%%%%%%%%%%%%%%%%%%%%%%%%%%%%%%%%%%%%%%%%%%%%%%%%%%%%%%%%%%%%%%%%%%%%%%%%

\section{Combinatorics of the Complexity Profile}
\label{Combinatorics of the Complexity Profile}
Previous works have developed and applied an explicit formula for the
complexity profile \cite{baryam2004a, baryam2004b, baryam2004c,
  sgs2004, metzler2005}.  This formula applies to the case that all
components have equal intrinsic scales.  To construct this formula, we
first define the quantity $Q(j)$ as the sum of the joint information
of all collections of $j$ components:
\begin{equation}
Q(j) = \sum_{i_1,\ldots,i_j} H(a_{i_1},\ldots,a_{i_j}).
\label{eq:Q}
\end{equation}
The complexity profile can then be expressed as
\begin{equation}
C(k) = \sum_{j = N-k}^{N-1} (-1)^{j+k-N} 
\binom{j}{j+k-N}
       Q(j+1),
\label{eq:CP}
\end{equation}
where $N=|A|$ is the number of components in $\cA$ \cite{baryam2004b,
  baryam2004c}.  The coefficients in this formula can be inferred from the
inclusion-exclusion principle \cite{erickson1996}.  Equation~\eqref{eq:CP} provides a method for
computing the complexity profile for any system from the values $H(U)$ of the information function.

To relate Eq.~(\ref{eq:CP}) to the properties of the complexity profile discussed in
Section~\ref{sec:ComplexityProfile}, we consider an arbitrary system $\cA$ of three components, $A= \{a,b,c\}$.  At scale $k=1$, Eq.~\eqref{eq:CP} gives
\begin{equation}
C(1) = Q(3) = H(a,b,c).
\end{equation}
We note that $C(1)$ equals the total information in $\cA$, consistent with Property 2 of Section \ref{sec:ComplexityProfile}.  At scale 2,
\begin{align*}
C(2) & = Q(2)-2Q(3)\\
& = H(a,b) + H(a,c) + H(b,c) - 2H(a,b,c)\\
& = I(a;b|c) + I(a;c|b) + I(b;c|a) + I(a;b;c).
\end{align*}
From the final expression above, it can be seen that $C(2)$ equals the total information in all dependencies
of scale 2 or higher.  We observe that this quantity vanishes if all variables are independent.  Finally, for scale 3,
\begin{align}
C(3) & = Q(1) - Q(2) + Q(3)\nonumber\\
& = H(a) + H(b) + H(c) - H(a,b) - H(a,c) - H(b,c) + H(a,b,c)\nonumber\\
& = I(a;b;c).\label{eq:C3}
\end{align}
Thus $C(3)$ returns only the information that is shared among all
three variables (which may be negative, as in running example $\exD$).
Again, for independent variables, $C(3)$ vanishes.

%%%%%%%%%%%%%%%%%%%%%%%%%%%%%%%%%%%%%%%%%%%%%%%%%%%%%%%%%%%%%%%%%%%%%%%%%%%%%

\section{Special Classes of Systems}
\label{Special Classes of Systems}

\subsection{Independent collection of intradependent blocks}
\label{Independent Subsystems}

One important special class of systems is those which break down into
independent subsystems (``blocks'') such that all components within
each block are entirely interdependent.  Examples \exA{}, \exB{} and
\exC{} all have this property.  In example \exC{}, components $a$ and
$b$ are in one block and component $c$ is in another.  For such
systems, the complexity profile and the marginal utility of
information can both be easily computed and are related to each other
in a simple manner.

In the simplest case, the entire system comprises a single block. Example \exB{} is such a system, in that the state of any one bit determines the state of all bits.   More generally, for any system $\cA=(A, H_\cA)$ which comprises a single block, each nonempty subset of components contains complete information about the system: $H(V)=H(A)$ for all nonempty $V \subset A$.  Using the definition of the complexity profile, we find that $C(x)$ has constant value $H(A)$ for all $0 \leq x \leq \sigma$ and is zero for $x>\sigma$, where $\sigma$ is the total scale of all components.  We can express $C(x)$ using a step function:
\begin{equation}
\label{eq:CPoneblock}
C(x) = H(A) \Theta(\sigma-x),
\end{equation}
where the $\Theta(y)$ has value 1 for $y \geq 0$ and 0 otherwise.

To compute the marginal utility of information for such a system, we observe that a descriptor with maximal utility will have $I(d;V)= \min\{y, H(A)\}$ for each subset $V \subset A$ and each value of the informational constraint $y$.  From this it follows that 
\begin{equation}
\label{eq:MUIoneblock}
\mui(y) = \sigma \Theta \big(H(A)-y \big).
\end{equation}
We observe that the reflection (generalized inverse; see Section \ref{sec:MUI}) $\tilde{\mui}(x)$ of $\mui(x)$ coincides exactly
with $C(x)$.

More generally, we can consider a system which comprises $m$ independent blocks.  A block is defined as a subsystem $\cB = (B, H_\cB)$ with the property that $H(V)=H(B_i)$ for each nonempty $V \subset B_i$.  Suppose $\cA$ is the disjoint union of blocks $\cB_i = (B_i, H_{\cB_i})$ for $i=1, \ldots, m$ which are independent as subsystems (see Section \ref{sec:Independence}).  Then additivity over independent subsystems (Property 4 in Sections \ref{sec:ComplexityProfile} and \ref{sec:MUI}), together with Eqs.~\eqref{eq:CPoneblock} and \eqref{eq:MUIoneblock}, implies that
\begin{equation}
C(x) = \tilde{\mui}(x) = \sum_{i=1}^m H(B_i) \Theta(\sigma_i-x),
\end{equation}
where $\sigma_i$ is total scale of components in block $\cB_i$.

We have thus established the following \emph{reflection principle} for systems of this type:
\begin{theorem}
\label{reflection}
For any system $\cA$ composed of independent blocks, the complexity profile and the MUI are reflections of each other:
\begin{equation}
C(x) = \tilde{\mui}(x).
\end{equation}
\end{theorem} 

This relationship between $C(x)$ and $M(y)$ does not hold for every system.  We show in Appendix~\ref{app:MUIpbit} that $C(x)$ and $M(y)$ are not reflections of each other in the case of example \exD, and, more generally, for a class of systems that exhibit negative information.  

\subsection{Systems with exchange symmetry among components}

We can simplify the equations for the complexity profile for
  systems which have exchange symmetry---all subsets having
  the same number of components contain the same amount of
  information.   Formally, for each set $U \subset A$,
the information of $U$ is a function of the cardinality $|U|$,  written as a subscript,
$H(U) = H_{|U|}$.  Examples $\exA$, $\exB$ and $\exD$
satisfy this constraint, but example $\exC$ does not.

The monotonicity axiom, defined in Section~\ref{Information},
implies that $H_k \leq H_{k+1}$.  Furthermore, the strong
subadditivity axiom, Eq.~(\ref{eq:strong-subadditivity}), implies that
if we take the sets $U = \{a,b\}$ and $V = \{b,c\}$, then
\begin{equation}
H_3 -H_2 \leq H_2 - H_1.
\end{equation}
It is easy to verify that this inequality holds for examples $\exA$,
$\exB$ and $\exD$.  For a symmetric system of $N$ components, we
have the more general ``concavity" property
\begin{equation}
H_{n+2} - H_{n+1} \leq H_{n+1} - H_n.
\end{equation}
This follows from considering the two overlapping sets $U =
\{a_1,\ldots,a_n,a_{n+1}\}$ and $V =
\{a_2,\ldots,a_n,a_{n+1},a_{n+2}\}$.  The symmetry condition lets us
write $H(U) = H(V) = H_{n+1}$, while the information of their union is
$H(U \cup V) = H_{n+2}$ and that of their intersection is $H(U \cap V)
= H_n$.  From this concavity property, it follows that if $H_{j+1} =
H_j$ for some $j$, then $H_k = H_j$ for all $k \in \{j,\ldots,N\}$;
that is, once the information levels off, it stays level.

Concavity is easy to verify if $H_j$ is constant, the case of complete
interdependence; or if $H_j$ is proportional to $j$, the case of complete
independence. It also is manifest in the more general situation $H_j \propto
j^\alpha$, where the ``independence parameter'' $\alpha$ interpolates
from $\alpha = 0$ (interdependence) to $\alpha = 1$ (independence).

Exchange symmetry also simplifies the form of the complexity
profile.  The result takes a particularly appealing form when stated
in terms of the information in dependencies of scale $k$ and
no higher, which we denote $D(k)$.  Recalling that the complexity
profile $C(k)$ indicates the information in dependencies of
scale $k$ and higher, we write
\begin{equation}
D(k) = C(k) - C(k+1).
\label{eq:D-of-k}
\end{equation}
The sum of $D(k)$ over all scales $k$ is $C(1)$.  As we did
for~$C(k)$, we can write a combinatorial formula for~$D(k)$:
\begin{equation}
D(k) = \sum_{j = N-k}^N (-1)^{j+k-N+1}
\binom{j}{j+k-N}
       Q(j).
\label{eq:D-of-k-combinatoric}
\end{equation}
When exchange symmetry holds, the information specific to
scale $k$ becomes
\begin{comment}
\begin{equation}
D(k) = \binom{N}{k} \sum_{j=N-k}^N (-1)^{j-N+k+1}
       \binom{k}{N-j} H_{j}.
\end{equation}
Likewise, the information of scale $k$ and higher becomes
\begin{equation}
C(k) = \binom{N}{k} \sum_{j = N - k}^{N-1} (-1)^{j+1-N}
       \frac{N-j}{j+1}
       \binom{k}{N-j} H_{j+1}.
\end{equation}
Shifting the summation index to $l \equiv j + k - N$,
\end{comment}
\begin{equation}
D(k) = \binom{N}{k} \sum_{l=0}^k (-1)^{l+1}
       \binom{k}{l} H_{l+N-k},
\end{equation}
while the information of scale $k$ and higher becomes
\begin{equation}
C(k) = \binom{N}{k} (-1)^{k}\sum_{l=0}^{k-1} (-1)^{l+1}
       \frac{k-l}{l+N-k+1}
       \binom{k}{l} H_{l+N-k+1}.
\end{equation}
For any fixed scale $k$, the complexity $D(k)$ is (up to a prefactor)
the binomial transform of the sequence $a_l \equiv H_{l+N-k}$.  This,
combined with the concavity property, allows
one to confirm that $D(k) \geq 0$ for~$k = 1,2$; \emph{i.e.,}
complexity can only be negative at scale $k = 3$ or higher.
Negative $D(k)$ arises from the leveling-off of the information
content $H_j$.

The binomial transform of a sequence can be rewritten using the
forward difference operator, $\diff$, whose action on a sequence
$\{a_n\}$ is given by~$(\diff a)_n = a_{n+1} - a_n$.  The complexity
$D(k)$ is given by the $k$\textsuperscript{th} finite difference
of~$H_{l+N-k}$:
\begin{equation}
D(k) = \binom{N}{k}(-1)^k (\diff^k H_{N-k})_0.
\end{equation}
Exchange symmetry among components is a reasonable and useful
simplification for some physical systems.  We discuss its relevance to
kinetic theory in Appendix~\ref{app:green}.  Previous work
  studied the complexity profile of the Ising model in the case of
  exchange symmetry \cite{sgs2004}.

\subsection{Weakly Interdependent Systems}
\label{kinetic theory}

Suppose that the components of our system are only weakly coupled, as
would be the case in a nearly-ideal gas or a magnet at high
temperature.  Then the complexity profile $C(k)$ will be rapidly
decaying, similar to example \exA, and the total scale-weighted
information of the dependency space, $S(\depspace)$, will be roughly
given by the first-scale complexity $C(1)$.  For some purposes, $C(1)$
is what we wish to obtain: for a physical system in thermal
equilibrium, $C(1)$ is the physical entropy, which connects statistics
to thermodynamics.  We now derive approximations for $S(\depspace)$
and for $C(1)$ which are useful in the weak-coupling limit.

From the conservation-law property of the complexity profile,
Eq.~(\ref{eq:cp-conservation}), we know that the total scale-weighted
information $S(\depspace)$ is the sum of $C(k)$ over all scales $k$.
Progressively improved approximations can be obtained by taking
partial sums of the form
\begin{equation}
S(\depspace) \approx \sum_{k=1}^{k_\mathrm{max}} C(k),
\label{eq:approx1}
\end{equation}
where $1 \leq k_\mathrm{max} \leq |A|$ is the degree of the
approximation.  This method is applicable when dependencies at larger
scales---binary, tertiary and so forth---become less significant even
as their number increases combinatorially.  The approach relies on
neglecting shared information at scales greater than a cutoff
$k_\mathrm{max}$, \emph{i.e.,}~large-scale dependencies among the
system components. In some circumstances, this approximation can
characterize the system behavior.

We now develop a systematic approach for approximating $C(1)$ given
quantities of shared information pertaining to progressively larger
scales.  For the first-order approximation, we neglect all shared
information pertaining to scales greater than 1, yielding
\begin{equation}
C(1) \approx S(\depspace)=\sum_{i=1}^{|A|} H(a_i).
\end{equation}
This is the first-order approximation according to
Eq.~\eqref{eq:approx1}.  We refine this approximation using the
inclusion-exclusion principle applied to the dependency space.  If
information is shared among pairs of components, the first-order
estimate of~$C(1)$ is too large.  We subtract from it the
shared information within pairwise dependencies.
\begin{equation}
C(1) \approx \sum_{i=1}^{|A|} H(a_i) - \sum_{i\neq j} I(a_i;a_j).
\end{equation}
This, in turn, undercounts the shared information content of tertiary
dependencies, so we add the tertiary mutual information summed over
all triplets, and so on.  Continuing this process, we write the
entropy $C(1)$ as the sum
\begin{equation}
\label{eq:C1series}
C(1) = \sum_{i=1}^{|A|} H(a_i) - \sum_{i\neq j} I(a_i;a_j)
       + \sum_{i\neq j\neq k} I(a_i;a_j;a_k) - \ldots
\end{equation}
Truncating this series after $k_\mathrm{max}$ terms, where $1 \leq
k_\mathrm{max} \leq |A|$, constitutes an approximation of the entropy
$C(1)$ to order $k_\mathrm{max}$.

If the system has exchange symmetry as discussed
in the previous section, then the shared information of any
dependency including $k$ components is
\begin{equation}
I_k \equiv \sum_{l=1}^k (-1)^{l+1}\binom{k}{l} H_l.
\end{equation}
With this relation, Eq.~\eqref{eq:C1series} for the joint entropy simplifies to
\begin{equation}
C(1) = \sum_{k=1}^{|A|} (-1)^{k+1}\binom{|A|}{k} I_k.
\label{eq:green-joint-info}
\end{equation}
The complete sum yields the exact value of~$C(1)$, which is the
  joint information of all components, $H_{|A|}$.  Indeed, if one
performs the entire sum over all scales, everything cancels except
$H_{|A|}$, because the binomial transform from~$\{H_l\}$
  to~$\{I_k\}$ is its own inverse.

One field where this approximation is valuable is the kinetic theory
of fluids \cite{green1952}.  Here, one is interested in approximating
the entropy $C(1)$ as well as possible given only small-scale
correlations.  Green's expansion is a method for doing this
systematically.  The terms in Green's expansion are integrals over
probability distributions involving successively larger numbers of
variables (see Appendix~\ref{app:green}).  However, the motivation for
each term, and the derivation of the coefficients, is not
straightforward.  The meaning of Green's expansion becomes clear when
the expansion is interpreted using Shannon information theory and our
multiscale formalism.  Green's expansion is
Eq.~(\ref{eq:green-joint-info}), written in the language of kinetic
theory.  Furthermore, all the coefficients in Green's entropy
expansion follow from the fact that the binomial transform is
self-inverse.  This is one example of the valuable perspective gained
by starting with a general axiomatic framework.

\section{Multiscale Cybernetic Thermodynamics}
\label{Multiscale Cybernetic Thermodynamics}
Thus far, we have considered system structure as an unchanging
quantity, and without explicit interaction of the system with its
environment.  We now build on this conceptual foundation by studying
systems influenced by their surroundings.  We consider the problem of
\emph{intentional} influences, which an agent outside a system uses to
regulate, guide or exploit that system.  Our approach enables us to
consider one of the primary limitations which intentional agents often
face.  Typically, an agent has only partial information about a system
of interest.  Furthermore, the available information may pertain to a
limited set of scales.  Our multiscale formalism allows us to express
the limitations which an agent faces in such a situation.

We consider, as a simple but illustrative example, the \emph{Szil\'ard
  engine,} a \emph{gedankenexperiment} consisting of a cylinder
immersed in a heat bath \cite{bennett1982, feynman1996, delrio2010,
  toyabe2010, koski2014, jun2014}.  Each end of the cylinder (left and
right) is a moveable piston.  In the middle of the cylinder is a
partition separating the left and right halves which can be removed
and reinserted, and somewhere within the cylinder, on one side or the
other of the partition, is a single atom.  When the Szil\'ard engine
is in thermal equilibrium with the surrounding heat bath, we can
extract useful work from it, provided we know which side of the
partition the atom is on.

The operational cycle of the Szil\'ard engine extracts energy from
information.  The engine operator (engineer) uses one bit of knowledge
about the atom's location, which side of the partition it is on, to
extract an energy $k_BT\log 2$.  After the operation, the atom is
equally likely to be on either side of the partition, so further
cycling requires gaining new knowledge about the engine's internal
configuration (and, if the engineer has a finite memory, therefore
requires erasing the prior datum within that memory
\cite{feynman1996}).  An engineer who has no knowledge of the atom's
position inside the Szil\'ard cylinder is just as likely to
\emph{expend} energy working the machine as they are to extract it, so
on average, they will obtain no useful work from the device.

The process of energy extraction from information starts with the
partition in place and engineer knowledge of which side of the
partition the atom is on.  If the atom is on the left side of the
partition, the engineer pushes the piston in from the right-hand side
without expending energy. The engineer removes the partition and the
bouncing atom pushes the cylinder back as heat flows into the cylinder
from the reservoir.  The heat flow keeps the atom at the same average
kinetic energy despite pushing the cylinder.  After the piston reaches
the right-hand end of the cylinder, the engineer re-inserts the
partition.  At this time, the atom can be anywhere within the
cylinder.  The magnitude of the energy obtained $k_BT\log 2$ is set by
the thermal energy of the heat bath, $k_B T$.  The factor of~$\log 2$
originates from the change in the spatial volume accessible to the
atom during the Szil\'ard engine cycle, which doubles.  A doubling in
volume is associated with an increase of thermodynamic entropy given
by
\begin{equation}
\Delta S = k_B \log\left(\frac{V_{\rm final}}{V_{\rm initial}}\right)
 = k_B \log 2.
\end{equation}

The information resource required to operate a Szil\'ard engine is
more properly expressed as a \emph{mutual information} between the
engine and its engineer (or control mechanism).  Consider an engineer
presented with an ensemble of $L$ Szil\'ard cylinders.  If the
configuration of each cylinder is predictable, then the engineer can
extract $L k_B T\log 2$ of energy by the end of the sequence.  If the
configurations are completely uncertain, the expected energy gain
averages to zero.  More generally, the energy gain decreases by $k_B T
\log 2$ for each cylinder for which the engineer must ask, ``Is the
atom on the left side of the partition?''  Thus, the amount of energy
which the engineer can extract from this sequence is $(L - H) k_B T
\log 2$, where $H$ is the number of yes/no questions which the
engineer must ask about the sequence~\cite{bennett1982, feynman1996}.
The number of yes/no questions about~$X$ which one can answer knowing
the value of $Y$ is their mutual information.  If the engineer has
access to a variable $Y$ which provides partial information about the
configuration of the cylinder sequence, then $H$ is reduced by the
mutual information between $Y$ and the cylinders, and the energy gain
increases proportionally.  Since this is true for an ensemble of
independent cylinders, for each cylinder in the ensemble the expected
energy gain is proportional to the available information about that
cylinder.

We can also consider multiple Szil\'ard cylinders as a single
system, which leads to a multiscale generalization.  Imagine $N$
Szil\'ard cylinders immersed in a heat bath at temperature $T$.  The
relevant property of each cylinder, the side occupied by an atom, is a
random variable.  Knowing about the positions of the atoms inside the
cylinders---that is, having a description of the $N$ system
components---allows an engineer to extract energy, at the cost of
making obsolete that knowledge.  Correlations among cylinders imply
that knowledge applicable to one is also applicable to another, so
that knowledge of one cylinder can be leveraged for a greater energy
gain.

When we characterize the configuration of a multi-cylinder Szil\'ard
engine, a natural measure of the usefulness of a descriptor is the
amount of energy we can extract from the machine using that
descriptor.  Here, the benefit of having a formalism that
characterizes the mutual information between the observer and the
system becomes apparent.  The available energy is proportional to the
utility defined in Section~\ref{sec:MUI}.  The descriptor $d$ has a
mutual information $I(d;X_i)$ with the $i$\textsuperscript{th}
cylinder of the engine.  Having this much information about cylinder
$X_i$ enables extracting from~$X_i$ a quantity of energy proportional
to the mutual information and to the thermal energy $k_BT$.  (Sagawa
and Uedo~\cite{Sagawa2012} provide an explicit protocol for extracting
the energy $k_B T I(d;X) \log 2$ in the case where the descriptor $d$
provides accurate knowledge of the cylinder $X$ with some error rate,
$\epsilon$.  The key step of the protocol is to only move the piston
partway, due to the probability of error.)

The MUI measures the amount of additional energy which can be gained
by making use of additional information. Given the ability to
choose information that one knows about the system, the additional
energy that can be gained is $\Delta E = M(y) k_BT \log 2$.

A real-world engineer working with ordinary tools can possess only
coarse-grained information about a system.  Therefore, what the
engineer can do with that system is limited.  Classical thermodynamics
is a phenomenological macroscopic treatment of this situation.  The
other extreme is the hypothetical being known as Maxwell's Demon,
which has exhaustive information about the finest-scale details of the
system.  The demon can exploit this information to extract the maximal
possible energy.  Descriptions having partial utility realize the
``continuum of positions'' \cite{fuchs2012inbook} between these two
extremes.

We can use our indices of multiscale structure to characterize what an
intentional agent can do when equipped with information that applies
to particular scales.  A single bit that is relevant at a large scale
provides the opportunity to extract a large amount of energy.  For
example, given $k$ dependent cylinders, we can extract in total
$k \times k_B T \log 2$ units of energy by acting independently on
each cylinder.  There are subtleties, however, in the macroscopic
process of extracting this energy.  If the information that is
available indicates that all cylinders are in the same state, a single
coherent action may be used to extract all the energy.  If the
cylinders are specified to alternate in some spatially structured way,
the ability to extract the energy using a coherent action requires a
mechanism to couple to that alternating structure.

More generally, we can consider engines that comprise independent
blocks of cylinders.  A multi-cylinder Szil\'ard engine of this type
is a system in which all components have the same intrinsic scale and
one bit of information apiece: $\significance(a) = \sigma$, $H(a) = 1$
for all $a \in A$.  Then $D(k)$, as defined in Eq.~(\ref{eq:D-of-k}),
is the number of blocks of size $k$.  Knowing the internal
configuration of each block requires one bit of information and
enables the extraction of $(k_BT)k\log 2$ in energy.  One block is not
correlated with another, so making use of a second block requires a
second bit of data.  In all, making use of all blocks at scale $k$
requires $D(k)$ bits and results in an energy gain $E$ given by
\begin{equation}
\frac{E}{k_BT} = D(k)k\log 2.
\end{equation}
We recall that generally the sum of $D(k)$ over all scales $k$ is
$C(1)$, which in this context is the joint Shannon information for the
entire multi-cylinder Szil\'ard engine.  Therefore, $C(1)$ is the
amount of information required in order to extract the energy from all
the blocks.

For any multi-cylinder Szil\'ard engine, even one not made of
independent blocks, if we have $C(1)$ bits of information, we can
predict the configuration of all the cylinders.  We can, therefore,
extract the maximum total amount of energy, by operating on each
cylinder in turn.  However, it is \emph{not} generally true that
$D(k)$ represents an extractable amount of energy for each value
of~$k$, even though summing $D(k)$ over all $k$ always yields $C(1)$.
If $D(k)$ is negative for some scale $k$, as in example \exD, then
there exists \emph{no} partial description which allows the Szil\'ard
engine operator to extract the energy associated with scale $k$ and no
other.  Information which can only be specified implicitly cannot be
utilized in isolation, only as part of an operation on a larger
dependency within a system.

\section{Discussion}
\label{Discussion}

\subsection{Characterizing complex systems}

Let us return to the question posed in the Introduction of how a complex system can be quantitatively defined.  Of all systems of $n$ components, with fixed values $H(a_1), \ldots, H(a_n)$ for the information of individual components,  which can be characterized as ``complex'' and what constitutes an appropriate measure of complexity?  The maximal total information $H(A)$ is achieved by letting all components be independent, so that $H(A) = H(a_1) + \ldots + H(a_n)$.  However, such a system contains no nontrivial interactions or dependencies, and is thus rather simple from a complex systems point of view.

Our formalism resolves this difficulty by emphasizing that all information applies at a particular scale.  In a system of fully independent components, information is maximized at the lowest scale but is absent at any higher scale.  In contrast, the systems of greatest interest to complex systems researchers contain information at many scales, with larger-scale information arising from redundancy in smaller-scale information.  This key property of complex systems is captured in our two indices of structure, the complexity profile and the marginal utility of information.  Both indices quantify the amount of information that applies at each scale, allowing the systems that exhibit multiscale complexity to be identified.

To illustrate this point, consider the example mentioned in the
Introduction of a box containing both a crystal and an ideal gas.  For
this system, information applies at two scales: that of the crystal
and that of the gas particles.  The complexity profile for the
contents of the box is the sum of two rectangles (\emph{i.e.,}~step
functions), one indicating the large-scale structure of the crystal
and the other the small-scale structure of the ideal gas.  By the
reflection principle, the MUI curve for this joint system is also the
sum of two rectangles.  Both indices of structure make clear that the
gas-and-crystal example lacks the multiscale organization that
distinguishes complex systems.

All systems are subject to a tradeoff in independence versus interdependence, due to the fact that larger-scale information arises from overlaps in the information pertaining to indvidual components.  This tradeoff is captured in our formalism by the conservation of the total scale-weighted information $S(\depspace_\cA)$ (Theorem \ref{totalSI}).  Both the complexity profile and the MUI reflect this tradeoff in their respective conservation laws, Eqs.~\eqref{eq:cp-conservation} and \eqref{eq:mui-conservation}.

\subsection{Negentropy}

The idea of using entropy or information to quantify structure has deep roots in physics.  One of the earliest and most influential attempts was Schr\"{o}dinger's concept of \emph{negative entropy} \cite{schrodinger1944life}, later shortened to \emph{negentropy} \cite{brillouin1953negentropy}, defined as the difference between a system's actual entropy and the maximum possible entropy of a system with the same matter, energy and volume.  Schr\"{o}dinger introduced negentropy to express the quality of order in living organisms due to their nonequilibrium nature.  Living beings are not in a state of maximum possible entropy, and negentropy is an attempt to quantify this difference.  We consider this attempt to be limited, in that it does not capture the multiscale aspect of organization present in living systems.

To make this point clear, we first define negentropy exactly within our information-theoretic framework.  Consider a physical system $\cA$ in which the information of a subset of components is defined as the joint physical entropy of these components considered together.  For convenience, assume each component has unit scale.  The maximum possible entropy of such a system---which would be attained if all components of $\cA$ were independent---is equal by Theorem \ref{totalSI} to the total scale-weighted information $S(\depspace_\cA)$.  In contrast, the actual entropy of $\cA$ is equal to the total (non-scale-weighted) information $I(\depspace_\cA)$, which can also be identified as $C(1)$.  The negentropy (\emph{i.e.,} the difference between the maximum possible and actual entropy) can thus be defined as $J(\cA) = S(\depspace_\cA)-I(\depspace_\cA)$.  We remark that negentropy is equivalent to the quantity called ``multi-information'' in network information theory \cite{studeny1998multiinformation,schneidman2003network,kolchinsky2011}. Since we have assumed that all components have scale one, the total information, $I(\depspace_\cA)$, equals the information at scale one, $C(1)$.  We can then use the conservation law, Eq.~\eqref{eq:cp-conservation}, to express negentropy in terms of the complexity profile:
\begin{equation}
J(\cA) = \sum_{k=1}^\infty C(k) - C(1) = \sum_{k=2}^\infty C(k).
\end{equation}

Negentropy represents a limited view of organization in that it treats
as the same all scales but the smallest.  According to this measure,
it is irrelevant whether a decrease of entropy arises from many bits
at scale 2 or a few bits at much larger scale.  This can be seen from
examples \exC{} and \exD{}, which both have negentropy equal to one
bit, despite having qualitatively different kinds of structure.  The
importance of the specific scale of the structure is captured by our
indices that include scale as a complementary axis to information.
Given the significance of macroscopic structure to scientific
observations of physical, biological and social systems, it seems that
a useful measure must necessarily make this distinction.

\subsection{Requisite variety}

The discipline of cybernetics, an ancestor to modern control theory,
used Shannon's information theory to quantify the difficulty of
performing tasks, a topic of relevance both to organismal survival in
biology and to system regulation in engineering.  Cyberneticist
W.\ Ross Ashby considered scenarios in which a regulator device must
protect some important entity from the outside environment and its
disruptive influences~\cite{ashby1956}.  In Ashby's examples, each
state of the environment must be matched by a state of the regulatory
system in order for it to be able to counter the environment's
influence on a protected component.  Successful regulation implies
that if one knows only the state of the protected component, one
cannot deduce the environmental influences; \emph{i.e.,} the job of
the regulator is to minimize mutual information between the protected
component and the environment.  This is an information-theoretic
statement of the idea of homeostasis.  Ashby's ``Law of Requisite
Variety'' states that the regulator's effectiveness is limited by its
own information content, or \emph{variety} in cybernetic terminology.
An insufficiently flexible regulator will not be able to cope with the
environmental variability.  A multiscale extension of Shannon
information theory provides a multiscale cybernetics, with which we
can study the scenarios in which ``that which we wish to protect'' and
``that which we must guard against'' are each systems of many
components, as are the tools we employ for regulation and
control~\cite{baryam2004a, baryam2004b, baryam2004c}.

Multiscale information theory enables us to overcome a key limitation
of the requisite variety concept.  In the examples of traditional
cybernetics~\cite{ashby1956}, each action of the environment requires
a specific, unique reaction on the part of the regulator.  This
neglects the fact that the impact which an event in the environment
has on the system depends upon the \emph{scale} of the environmental
degrees of freedom involved.  There is a great difference between
large-scale and fine-scale impacts.  Systems can deflect fine-scale
impacts without needing to specifically respond to them, while they
need to respond to large-scale ones or perish.  For example, a human
being can be indifferent to the impact of a falling raindrop, whereas
the impact of a falling rock is much more difficult to neglect, even
if specifying the state of the raindrop and the state of the rock
require the same amount of information.  An extreme case is the impact
of a molecule: air molecules are continually colliding with us, yet
the only effects we have to cope with actively are the large-scale,
collective behaviors like high-speed winds.  Ashby's Law does not make
this distinction.  Indeed, there is no framework for the discussion
due to the absence of a concept of scale in the information theory he
used: Each state is equally different from every other state and
actions must be made differently for each different environment.

Thus, in order to account for the real-world conditions, a multiscale
generalization of Ashby's Law is needed. According to such a Law, the
responses of the system must occur at a scale appropriate to the
environmental change, with larger-scale environmental changes being
met by larger-scale responses.  As with the case of raindrops
colliding with a surface, large-scale structures of a system can avoid
responding dynamically to small-scale environmental changes which
cause only small-scale fluctuations in the system.

Given a need to respond to larger-scale changes of the environment,
coarser-scale descriptions of that environment may suffice.  A
regulator that can marshall a \emph{large-scale response} can use a
coarse-grained description of the environment to counteract
large-scale fluctuations in the external conditions.  In this way,
limited amounts of information can still be useful.  To make requisite
variety a practical principle, one must recognize that information
applies to specific scales.

Ashby aimed to apply the requisite variety concept to biological
systems, as well as technological ones.  An organism which lacks the
flexibility to cope with variations in its environment dies.  Thus, a
mismatch in variety/complexity is costly in the struggle for survival,
and so we expect that natural selection will lead to organisms whose
complexity matches that of their environment.  However, ``the
environment'' of a living being includes other organisms, both of the
same species and of others.  Organisms can act and react in concert
with their conspecifics, and the effect of any action taken can depend
on what other organisms are doing at the same time~\cite{allen2014}.
In some species, such as social insects~\cite{tschinkel2014}, distinct
scales of the individual, colony and species are key features
characterizing collective action.  This suggests a multiscale
cybernetics approach to the evolution of social behavior: We expect
that scales of organization within a population---the scales, for
example, of groups or colonies---will evolve to match the scales of
the challenges which the environment presents.  Furthermore, the
concept of multiscale response applies within the individual organism
as well.  Multiple scales of environmental challenges are met by
different scales of system responses.  To protect against infection,
for example, organisms have physical barriers (\emph{e.g.,} skin),
generic physiological responses (\emph{e.g.,} clotting, inflammation)
and highly specific adaptive immune responses, involving interactions
among many cell types, evolved to identify pathogens at the molecular
level. The evolution of immune systems is the evolution of separate
large- and small-scale countermeasures to threats, enabled by
biological mechanisms for information transmission and
preservation~\cite{stacey2008}.  As another example, the muscular
system includes both large and small muscles, comprising different
numbers of cells, corresponding to different scales of environmental
challenge (\emph{e.g.,}~ pursuing prey and escaping from predators
versus chewing food)~\cite{bar2003complexity}.

\subsection{Benefits of an axiomatic formalism}

Because complex systems arise in a wide range of scientific contexts,
it is challenging to formulate consistent definitions for key
concepts.  Rooting our definitions in mathematically general axioms
for information enables our formalism to apply to a wide range of
empirical and model systems.

The axiomatic basis of our constructions also allows the unification
of ideas from different areas of mathematics and science.  For
example, the definition of independence in
Eq.~\eqref{eq:independence}---which subsumes the definitions of
independent random variables and of linearly independent vector
spaces---allows the notion of independence to be applied rigorously in
any context for which an information function is available.  A further
example is conditional independence, which is defined in probability
theory in terms of the joint distribution of three or more random
variables.  This definition implies a condition on the Shannon
entropies of the distributions involved, a condition which can be
abstracted to a more general information measure context.
Specifically, system components $a_1$ and $a_3$ are conditionally
independent given component $a_2$ if
\begin{equation}
H(a_1,a_3|a_2) = H(a_1|a_2) + H(a_3|a_2).
\label{eq:conditional-independence}
\end{equation}
This enables discussions of Markov chains, Markov random fields
\cite{steudel2010} and ``computational mechanics'' \cite{ellison2009,
  mahoney2009, crutchfield2009, crutchfield2011} to be subsumed in a
general formalism and thence applied in algorithmic, vector-spatial or
matroidal contexts.

\subsection{Approximations to the marginal utility of information}

Our new index of structure, the MUI, is philosophically similar to
data-reduction or dimensional reduction techniques like principal
component analysis, multidimensional scaling and detrended fluctuation
analysis \cite{peng1994,hu2002}; to the Information Bottleneck methods
of Shannon information theory \cite{slonim1999, shalizi2000,
  tishby2000, ziv2005}; to Kolmogorov structure functions and
algorithmic statistics in Turing-machine-based complexity theory
\cite{vereshchagin2003, gruenwald2004, vitanyi2006}; and to Gell-Mann
and Lloyd's ``effective complexity'' \cite{gellmann1996infomeasures}.
All of these methods are mathematical techniques for characterizing
the most important behaviors of the system under study.  Each is an
implementation of the idea of finding the best possible brief
description of a system, where description length is measured in bits
or by the number of coordinates employed.  However, MUI can be
formulated completely generally, in terms of our basic postulates for
information functions.  Furthermore, the MUI is by definition the
optimal such characterization.

We have defined MUI in terms of optimally effective descriptors: for
each possible amount of information invested in describing the system,
we use the descriptor which provides the best possible theoretical
return (in terms of scale-weighted information) on that investment.
However, in applied contexts, it may be difficult or impossible to
realize these theoretical maxima, due to constraints beyond those
imposed by the axioms of information functions.  It is often useful in
these contexts to consider a particular ``description scheme'', in which
descriptors are restricted to be of a particular form.  In this case,
the maximal utility we find following that description scheme can be
less than the theoretical maximal utility defined by the system's
dependency space.  We would in such a case find an approximation to
the MUI, rather than the MUI itself.

We can illustrate this issue with a straightforward description
scheme: using subsets of the system's component set as descriptors.
In this scheme, a descriptor $d$ is a set drawn from the set of all
components $A$, and the length of a descriptor is just the number of
components used in it.  Any descriptor $d$ naturally provides the full
quantity of information about the components from~$A$ which are
included in~$d$ itself.  If correlations exist among the system's
components, then the elements of~$d$ also provide information about
other components of the system.  The maximal utility \emph{possible
  within this description scheme} at a given descriptor length $x$ is
given by maximizing $\sum_{a \in A} I(d;a)$ over all possible choices
of descriptor $d$ which have length $x$.  The finite difference of
this utility curve is an approximation to the MUI.  It is not
difficult to see that for examples \exA, \exB\ and \exC\ this
approximate MUI is the same as the ideal MUI.  For example \exD,
however, the approximation and the ideal differ.  Because the
components in example \exD\ are pairwise independent, any one-component
descriptor only describes the component it mirrors, so the maximal
utility approximation at length 1 is 1.  A descriptor using two
components can describe the whole system, so the maximal utility
approximation at length 2 is 3.  Therefore, the approximated MUI in
the confines of this description scheme is a \emph{nonmonotonic}
function of descriptor length, starting at~1 and rising to~2 before
falling back to~0.  This nonmonotonicity is a consequence of the
negative information in example \exD's dependency space.
        
Other choices of description scheme are also possible.  These
description schemes bring other, more familiar quantities into the
information-theoretic framework.  Within an algorithmic context, for
example, one might study the increasing utility of algorithmic
descriptions as a function of the computational resources available.

\subsection{Mechanistic versus informational dependencies}

Our indices of structure measure multi-component relationships,
including statistical correlations among random variables.  A key
question is how causal interactions give rise to such relationships.
Importantly, causal influences at one scale can produce correlations
at another.  For example, the interactions in an Ising spin system are
pairwise in character: the interaction energy of a pair of spins is
not affected by other spins being up or down elsewhere in the lattice.
These pairwise couplings can, however, give rise to long-range
patterns \cite{kardar2007}.  Similarly, in commonly-used models of
coupled oscillators, the effect one oscillator has on another---the
force with which component $i$ pulls on component $j$---depends only
on the relative phase difference between those two oscillators, and
the total influence on an oscillator due to all the others with which
it is coupled is just the sum of their influences.  Yet, even though
the forces are dyadic, synchronization among oscillators creates
collective, coherent behavior \cite{dorogovtsev2010}.
Synchronization, in other words, creates structure on a large scale.

\subsection{Limitations of network representations}
Representing system structure by networks, a common practice in the
complex-systems field, prioritizes pairwise (scale two) relationships
and may neglect higher-scale dependencies.  Often, a network model is
formulated by computing a measure of correlation for each pair of
components, and drawing an edge between the corresponding vertices if
that correlation is statistically significant~\cite{dorogovtsev2010,
  harmon2010}.  This procedure discards information at scales three
and higher.

One way to incorporate higher-scale information into a network
representation is by changing what the presence of an edge means.  For
an illustrative example, take the case of gene regulation, in which
the system components are genes and a joint probability distribution
describes their expression levels.  We expect that if the expression
level of one gene predicts the that of another, this relationship may
be biologically significant.  We could make a straightforward network
depiction by linking the vertices corresponding to genes $a_i$ and
$a_j$ if their mutual information is large.  However, not all of these
edges will represent direct paths of biochemical influence.  Suppose
that gene $a_1$ boosts the expression of gene $a_2$, which in turn
boosts the expression of gene $a_3$.  In this case, the measured
mutual information of each of the three pairs $(a_1,a_2)$, $(a_1,a_3)$
and $(a_2,a_3)$ could be large, even though there is no direct causal
link between $a_1$ and $a_3$.  To distinuguish between direct and
indirect relationships, one must test for the \emph{conditional
  independence} of components that are correlated at the pairwise
level.  This analysis requires information on three-fold and higher
correlations.  Bayesian networks \cite{friedman2000using} are one way
of incorporating such higher-order correlations into a network model.

In contrast, our framework incorporates multiscale information not as
a tool for refining pairwise relationships, but as an important aspect
of structure in its own right.  Such multiscale information can arise
from pairwise mechanisms (as discussed above), or from causal
relationships that are intrinsically of scale three or higher. For
example, suppose the expression levels of three genes jointly satisfy
a nonlinear constraint due to joint reliance on a common precursor or
other resource.  Such a relationship may not be representable within a
Bayesian network, but can be represented in terms of information
shared among these genes.

\section{Conclusion}

Over the past century, science has made enormous strides in understanding the fundamental building blocks of physics and biology.  However, it is increasingly clear that understanding the behaviors of physical, biological and social systems requires more than a characterization of their constituent parts.  Rather, scientific progress depends on a theory of system structure.  While many conceptual elements of such a theory have been developed within the field of complex systems, a general quantitative framework has so far been elusive.

Our work aims to provide a mathematical foundation for complex systems theory, in which the fundamental concepts of dependence, scale and structure are given precise meaning.  This is achieved via an axiomatic formalism for information that generalizes classical information theory.  This formalism enables us to identify structure as the sharing of information among system components.  A system's structure can be summarized by its complexity profile or its MUI function, both of which highlight the scale of system behaviors.

Already, we have found that our framework resolves key conceptual puzzles, from the combinatorial origins of kinetic-theory expressions to the characterization of ``complex systems.''  We hope this mathematical formalism of structure can aid in the scientific transition from understanding the components of systems to understanding systems themselves.

\appendix

\section{Total scale-weighted information}
\label{app:SI}

Here we prove Theorem \ref{totalSI} of the main text, which we restate here for convenience:

\begin{thm1}
For any system $\cA$, the total scale-weighted information, $S(\depspace_\cA) = \sum_{x \in \depspace_\cA} s(x) I(x)$,  is given by the scale and information of each component, independent of the information shared among them:
\begin{equation}
S(\depspace_\cA) = \sum_{a \in A} \significance(a) H(a).
\end{equation}
\end{thm1} 

\begin{proof}
The proof amounts to a rearrangement of summations.  We begin with the definition of scale-weighted information,
\begin{equation}
S(\depspace_\cA) = \sum_{x \in \depspace_\cA} \scorder(x) I(x).
\end{equation}
Substituting the definition of $s(x)$ (main text, \ref{eq:scaledef}) and rearranging yields
\begin{align*}
S(\depspace_\cA)  
&= \sum_{x \in \depspace_\cA} \left( \sum_{\substack{a \in A \\ x \text{ includes }a}} \sigma(a) \right) I(x)\\
& = \sum_{a\in A}   \sigma(a) \sum_{\substack{x \in \depspace_\cA \\ x \text{ includes }a}} I(x)\\
& = \sum_{a \in A}  \sigma(a) I(\delta_a)\\
& = \sum_{a \in A}  \sigma(a) H(a). \qedhere
\end{align*}
\end{proof}

\section{Conservation Law for the Complexity Profile}
\label{app:conservation}
In this Appendix, we prove the conservation law for the complexity profile, Eq.~\eqref{eq:cp-conservation} of the main text. We state this law as follows:

\begin{theorem}
The area under the complexity profile of a system $\cA$ is equal to
the total scale-weighted information of $\cA$:
\begin{equation}
\int_0^\infty C(y) \; dy = S(\depspace_\cA).
\end{equation}
\end{theorem}

\begin{proof}
We begin by substituting the definition of $C(y)$:
\begin{align*}
\int_0^\infty C(y) \; dy &= \int_0^\infty  I \big(\{x \in \depspace_\cA\; : \; \significance(x) \geq y\} \big) \; dy\\
& = \int_0^\infty \left( \sum_{\substack{x \in \depspace_\cA \\ y \leq \significance(x)}} I(x) \right) dy.
\end{align*}
We then interchange the sum and integral on the right-hand side and apply Theorem \ref{totalSI}:
\begin{align*}
\int_0^\infty C(y) \; dy 
& = \sum_{x \in \depspace_\cA} \left( I(x) \int_0^{\significance(x)}  \; dy \right)\\
& = \sum_{x \in \depspace_\cA} \significance(x) I(x)\\
& = S(\depspace_\cA). \qedhere
\end{align*}
\end{proof}

\section{Properties of Independent Subsystems}
\label{app:independence}

Here we prove fundamental propeties of independent subsystems, which will be used to prove the additivity property of the complexity profile.  Our first target is the \emph{hereditary property of independence} (Theorem \ref{lem:indepsubsystems}), which asserts that subsystems of independent subsystems are independent \cite{perfect1981independence}.  We then establish in Theorem \ref{Icaseslem} a simple characterization of information in systems composed of independent subsystems.

For $i=1, \ldots, k$, let
$\cA_i=(A_i, H_{\cA_i})$ be subsystems of $\cA= (A,H_\cA)$, with the
subsets $A_i \subset A$ disjoint from each other.  We recall the information-theoretic definition of independent
subsystems from Section~\ref{sec:Independence}.

\begin{defn} The subsystems  $\cA_i=(A_i, H_{\cA_i})$ are \emph{independent} if 
\begin{equation*}
H(A_1 \cup \ldots \cup A_k) = H(A_1) + \ldots + H(A_k).
\end{equation*}
\end{defn}

We establish the hereditary property of independence first in the case of two subsystems (Lemma \ref{lem:twoindepsubsystems}), using repeated application of the strong subadditivity axiom. We then extend this result in Theorem \ref{lem:indepsubsystems} to arbitrary numbers of subsystems.

\begin{lemma}
If  $\cA_1$ and $\cA_2$ are independent subsystems of $\cA$, then for every pair of subsets $U_1 \subset A_1$, $U_2 \subset A_2$, $H(U_1 \cup U_2) = H(U_1) + H(U_2)$.
\label{lem:twoindepsubsystems}
\end{lemma}

\begin{proof}
The strong subadditivity axiom, applied to the sets $A_1$ and $U_1 \cup A_2$, yields
\begin{equation*}
H(A_1 \cup A_2) \leq H(A_1) + H(U_1 \cup A_2) - H(U_1).
\end{equation*}
Replacing the left-hand side by $H(A_1) + H(A_2)$ and adding $H(U_1)-H(A_1)$ to both sides yields
\begin{equation}
\label{U1A2}
H(U_1) + H(A_2) \leq H(U_1 \cup A_2).
\end{equation}
Now applying strong subadditivity to the sets $U_1 \cup U_2$ and $A_2$ yields
\begin{equation*}
H(U_1 \cup A_2) \leq H(U_1 \cup U_2) + H(A_2) - H(U_2).
\end{equation*}
Combining with \eqref{U1A2} via transitivity, we have
\begin{equation*}
H(U_1) + H(A_2)\leq H(U_1 \cup U_2) + H(A_2) - H(U_2).
\end{equation*}
Adding $H(U_2)-H(A_2)$ to both sides yields
\begin{equation}
\label{eq:U1U2leq}
H(U_1) + H(U_2) \leq H(U_1 \cup U_2).
\end{equation}
But strong subadditivity applied to $U_1$ and $U_2$ yields
\begin{equation}
H(U_1 \cup U_2) \leq H(U_1) + H(U_2) - H(U_1 \cap U_2)  \leq H(U_1) + H(U_2).
\label{eq:U1U2geq}
\end{equation}
We conclude from inequalities \eqref{eq:U1U2leq} and \eqref{eq:U1U2geq} that 
\begin{equation*}
H(U_1 \cup U_2) = H(U_1) + H(U_2). \qedhere
\end{equation*}
\end{proof}

We now use an induction argument to extend the hereditary property of independence to any number of subsystems.

\begin{theorem}
\label{lem:indepsubsystems}
If  $\cA_1, \ldots, \cA_k$ are independent subsystems of $\cA$, and $U_i \subset A_i$ for $i=1, \ldots, k$ then 
\begin{equation*}
H(U_1 \cup \ldots \cup U_k) = H(U_1) + \ldots + H(U_k).
\end{equation*}
\end{theorem}

\begin{proof}
This follows by induction on $k$.  The $k=1$ case is trivial.  Suppose inductively that
the statement is true for $k=\tilde{k}$, for some integer $\tilde{k}
\geq 1$, and consider the case $k=\tilde{k}+1$.  We have
\begin{equation*}
H(U_1) + \ldots + H(U_{\tilde{k}}) +H(U_{\tilde{k}+1}) = H( U_1 \cup \ldots \cup U_{\tilde{k}} )+H(U_{\tilde{k}+1}) 
\end{equation*}
by the inductive hypothesis, and 
\begin{equation*}
H( U_1 \cup \ldots \cup U_{\tilde{k}} )+H(U_{\tilde{k}+1})  = H( U_1 \cup \ldots \cup U_{\tilde{k}} \cup U_{\tilde{k}+1}) 
\end{equation*}
by Lemma \ref{lem:twoindepsubsystems} (since the subsystem of $\cA$ with
component set $A_1 \cup \ldots \cup A_{\tilde{k}}$ is clearly
independent from $\cA_{\tilde{k}+1}$).  This completes the proof.
\end{proof}

We now examine the information in dependencies for systems comprised of independent subsystems.  
For convenience, we introduce a new notion: The \emph{power system}
of a system $\cA$ is a system $2^\cA=(2^A, H_{2^\cA})$, where $2^A$ is
the set of all subsets of $A$ (which in set theory is called the
\emph{power set} of $A$).  In other words, the components of $2^\cA$ are the subsets of $A$.
The information function $H_{2^\cA}$ on $2^\cA$ is defined by the relation
\begin{equation}
H_{2^\cA} (U_1, \ldots, U_k ) = H_\cA(U_1\cup \ldots \cup U_k).
\end{equation}
By identifying the singleton subsets of $2^A$ with the elements of $A$
(that is, identifying each $\{a\} \in 2^A$ with $a \in A$), we can
view $\cA$ as a subsystem of $2^\cA$.

This new system allows us to use the following relation: For any
integers $k, \ell \geq 0$ and components $a_1, a_2, b_1, \ldots,
a_k, c_1, \ldots, c_\ell \in A$,
\begin{multline}
\label{mutinforelation}
I_{\cA} (a_1; a_2;  b_1; \ldots ; b_k | c_1, \ldots, c_\ell ) = 
I_{\cA} (a_1;  b_1; \ldots ; b_k | c_1, \ldots, c_\ell )\\
+ I_{\cA} (a_2;  b_1; \ldots ; b_k | c_1, \ldots, c_\ell )
- I_{2^\cA} (\{a_1, a_2\};  b_1; \ldots ; b_k | c_1, \ldots, c_\ell ).
\end{multline} 
This relation generalizes the identity
$I(a_1;a_2)=H(a_1)+H(a_2)-H(a_1,a_2)$ to conditional mutual information.  It follows directly from the
mathematical definition of $I$, Eq.~(\ref{eq:infosum}) of the main
text.

We now show that if $\cB$ and $\cC$ are independent subsystems of $\cA$, any conditional mutual information of components $\cB$ and components of $\cC$ is zero.

\begin{lemma}
\label{indepzerolem}
Let $\cB=(B, H_\cB)$ and $\cC=(C, H_\cC)$ be independent subsystems of
$\cA$.  For any components $b_1,
\ldots, b_m, b'_1, \ldots, b'_{m'} \in B$ and $c_1, \ldots, c_n, c_1',
\ldots, c'_{n'} \in C$, with $m,n \geq 1$, $m',n' \geq 0$,
\begin{equation}
I(b_1; \ldots; b_m;c_1; \ldots; c_n|b'_1, \ldots, b'_{m'}, c_1', \ldots, c'_{n'}) = 0.
\end{equation}
\end{lemma}

\begin{proof}
We prove this by induction.  As a base case, we take $m=n=1, m'=n'=0$.
In this case, the statement reduces to $I(b;c)=0$ for every $b \in B$,
$c \in C$.  Since Lemma \ref{lem:twoindepsubsystems} guarantees that
$H(b,c)=H(b)+H(c)$, this claim follows directly from the identity
$I(b;c)=H(b)+H(c)-H(b,c)$.

We now inductively assume that the claim is true for all independent
subsystems $\cB$ and $\cC$ of a system $\cA$, and all $m \leq
\tilde{m}, n \leq \tilde{n}, m' \leq \tilde{m}'$, and $n' \leq
\tilde{n}'$, for some integers $\tilde{m}, \tilde{n} \geq 1$,
$\tilde{m}', \tilde{n}' \geq 0$.  We show that the truth of the claim
is maintained when each of $\tilde{m}, \tilde{n}, \tilde{m}'$, and
$\tilde{n}'$ is incremented by one.

We begin by incrementing $m$ to $\tilde{m}+1$. Applying \eqref{mutinforelation} yields
\begin{multline}
\label{mutinfoinduct}
I_{\cA} \big(b_{\tilde{m}}; b_{\tilde{m}+1};  b_1; \ldots; b_{\tilde{m}-1}; c_1; \ldots; c_{\tilde{n}}|b'_1, \ldots, b'_{\tilde{m}'}, c_1', \ldots, c'_{\tilde{n}'} \big)\\ 
=I_{\cA} \big(b_{\tilde{m}}; b_1; \ldots; b_{\tilde{m}-1}; c_1; \ldots; c_{\tilde{n}}|b'_1, \ldots, b'_{\tilde{m}'}, c_1', \ldots, c'_{\tilde{n}'} \big)\\ 
+ I_{\cA} \big( b_{\tilde{m}+1};  b_1; \ldots; b_{\tilde{m}-1}; c_1; \ldots; c_{\tilde{n}}|b'_1, \ldots, b'_{\tilde{m}'}, c_1', \ldots, c'_{\tilde{n}'} \big)\\ 
- I_{2^\cA} \big(\{b_{\tilde{m}}; b_{\tilde{m}+1} \};  b_1; \ldots; b_{\tilde{m}-1}; c_1; \ldots; c_{\tilde{n}}|b'_1, \ldots, b'_{\tilde{m}'}, c_1', \ldots, c'_{\tilde{n}'} \big).
\end{multline} 
The first two terms of the right-hand side of \eqref{mutinfoinduct}
are zero by the inductive hypothesis.  Furthermore, it is clear from
the definition of a power system that $2^\cB$ and $2^\cC$ are
independent subsystems of $2^\cA$.  Thus the final term on the
right-hand size of \eqref{mutinfoinduct} is also zero by the inductive
hypothesis. In sum, the entire right-hand side of
\eqref{mutinfoinduct} is zero, and the left-hand side must therefore
be zero as well. This proves the claim is true for $m=\tilde{m}+1$.

We now increment $m'$ to $\tilde{m}'+1$.  From Eq.~\eqref{eq:Idependency} of the main text, we have the relation
\begin{multline*}
I_{\cA} \big(b_1; \ldots; b_{\tilde{m}}; c_1; \ldots; c_{\tilde{n}}|b'_1, \ldots, b'_{\tilde{m}'}, c_1', \ldots, c'_{\tilde{n}'} \big)\\ 
= I_{\cA} \big(b'_{\tilde{m}'+1}; b_1; \ldots; b_{\tilde{m}}; c_1; \ldots; c_{\tilde{n}}|b'_1, \ldots, b'_{\tilde{m}'}, c_1', \ldots, c'_{\tilde{n}'} \big)\\ 
+ I_{\cA} \big(b_1; \ldots; b_{\tilde{m}}; c_1; \ldots; c_{\tilde{n}}|b'_1, \ldots, b'_{\tilde{m}'}, b'_{\tilde{m}'+1}, c_1', \ldots, c'_{\tilde{n}'} \big).
\end{multline*}
The left-hand side above is zero by the inductive hypothesis, and the
first term on the right-hand side is zero by the case $m=\tilde{m}+1$
proven above.  Thus the second term on the right-hand side is also
zero, which proves the claim is true for $m'=\tilde{m}'+1$.

Finally, the cases $n=\tilde{n}+1$ and $n'=\tilde{n}'+1$ follow by
interchanging the roles of $\cB$ and $\cC$.  The result now follows by
induction.
\end{proof}

We next show that for $\cB$ and $\cC$ independent subsystems of $\cA$, the amounts of information in dependencies of $\cB$ are not affected by additionally conditioning on components of $\cC$.

\begin{lemma}
\label{indepexcludelem}
Let $\cB=(B, H_\cB)$ and $\cC=(C, H_\cC)$ be independent subsystems of
$\cA$.  For integers $m \geq 1$ and $m',n' \geq 0$, let $b_1, \ldots,
b_m \in B$, $c_1, \ldots, c_n, c_1', \ldots, c'_{n'} \in C$.  Then
\begin{equation}
I(b_1; \ldots; b_m|b'_1, \ldots, b'_{m'}, c_1', \ldots, c'_{n'}) = I(b_1; \ldots; b_m|b'_1, \ldots, b'_{m'}).
\end{equation}
\end{lemma}

\begin{proof}
This follows by induction on $n'$. The claim is trivially true
for $n'=0$.  Suppose it is true in the case $n'=\tilde{n}'$, for some $\tilde{n}' \geq 0$.
By Eq.~\eqref{eq:Idependency} we have
\begin{multline}
I(b_1; \ldots; b_m|b'_1, \ldots, b'_{m'}, c_1', \ldots, c'_{\tilde{n}'}) \\ 
= I(b_1; \ldots; b_m; c'_{\tilde{n}'+1}|b'_1, \ldots, b'_{m'}, c_1', \ldots, c'_{n'})\\
+ I(b_1; \ldots; b_m|b'_1, \ldots, b'_{m'}, c_1', \ldots, c'_{\tilde{n}'}, c'_{\tilde{n}'+1}).
\end{multline}
The left-hand side is equal to $I(b_1; \ldots; b_m|b'_1, \ldots,
b'_{m'})$ by the inductive hypothesis, and the first term on the
right-hand side is zero by Lemma \ref{indepzerolem}.  This completes
the proof.
\end{proof}

Finally, it follows from Lemmas \ref{indepzerolem} and \ref{indepexcludelem} that if $\cA$ separates into independent subsystems, an irreducible dependency of $\cA$ has nonzero information only if it includes components from only one of these subsystems.  To state this precisely, we introduce a projection mapping from irreducible dependencies of a system $\cA$ to those of a subsystem $\cB$ of $\cA$.  This mapping, denoted $\rho_\cB^\cA: \depspace_\cA \to \depspace_\cB$, takes an irreducible
dependency among the components in $A$, and ``forgets'' those
components that are not in $B$, leaving an irreducible dependency
among only the components in $B$.  For example, suppose $A = \{a, b,
c\}$ and $B = \{b,c\}$.  Then
\begin{align}
\rho_\cB^\cA (a;b|c) & = b|c \nonumber\\
\rho_\cB^\cA (b;c|a) & = b;c.
\end{align}

\begin{comment}
The preimage under $\rho_\cB^\cA$ of a dependency in $\cB$ ``decomposes'' this dependency into a set of irreducible dependencies in $\cA$. In the above example, we have
\begin{equation}
\left(\rho_\cB^\cA \right)^{-1}(b|c) = \big\{(a;b|c), \; (b|a,c) \big\} \subset \depspace_\cA.
\end{equation}
Information content is preserved under the preimage of $\rho_\cB^\cA$.  In the above example, this states that
\begin{equation}
I_\cB(b|c) = I_\cA (a;b|c) + I_\cA(b|a,c).
\end{equation}
\end{comment}

We can now state the following simple characterization of information in systems comprised of independent subsystems:

\begin{theorem}
\label{Icaseslem}
Let $\cA_1, \ldots, \cA_k$ be independent subsystems of $\cA$, with
$A=A_1 \cup \ldots \cup A_k$. Then for any irreducible dependency $x
\in \mathfrak{D}_\cA$,
\begin{equation}
I_\cA(x) = \begin{cases} 
I_{\cA_i} \big ( \rho^\cA_{\cA_i}(x) \big),
& \parbox{8cm}{if $x$ includes only components of $\cA_i$ \\for some $i \in \{1, \ldots, k\}$,}\\[5mm]
0 & \text{otherwise.} \end{cases}
\end{equation}
\end{theorem}

\begin{proof}
In the case that $x$ involves only components of $\cA_i$ for some $i$,
the statement follows from Lemma \ref{indepexcludelem}.  In all other
cases, the claim follows from Lemma \ref{indepzerolem}.
\end{proof}

\section{Additivity of the Complexity Profile}
\label{app:additivity}

Here we prove Property 4 of the complexity profile claimed in Section~\ref{sec:ComplexityProfile}: the
complexity profile is additive over independent systems.

\begin{theorem}
Let  $\cA_1, \ldots, \cA_k$ be independent subsystems of $\cA$.  Then 
\begin{equation}
C_\cA(y) = C_{\cA_1}(y) + \ldots + C_{\cA_k}(y).
\end{equation}
\end{theorem}

\begin{proof}
We start with the definition
\begin{equation}
C_{\cA}(y) = \sum_{\substack{x \in \depspace_\cA  \\\significance(x) \geq y}} I_\cA (x).
\end{equation}
Applying Theorem \ref{Icaseslem} to each term on the right-hand side yields
\begin{align*}
C_{\cA}(y) & = \sum_{i=1}^k \sum_{\substack{x \in \depspace_\cA \\ \text{$x$ includes only components of $\cA_i$} \\\significance (x) \geq y}} 
I_{\cA_i} \big( \rho^\cA_{\cA_i}(x) \big)\\
& = \sum_{i=1}^k \sum_{\substack{x \in \depspace_{\cA_i} \\\significance (x) \geq y}}  I_{\cA_i}  (x)\\
& = \sum_{i=1}^k C_{\cA_i}(y). \qedhere
\end{align*}
\end{proof}

\section{Additivity of Marginal Utility of Information}
\label{app:mui}

Here we prove the additivity property of MUI stated in Section \ref{sec:MUI}.  We begin by recalling the mathematical context for this result.  

The maximal utility of information, $U(y)$, is defined as the maximal value of the quantity
\begin{equation}
u = \sum_{a \in A} \sigma(a) I(d;a),
\end{equation}
as the variables in the set $\{I(d;V) \}_{V \subset A}$ vary subject to the following constraints:
\renewcommand{\labelenumi}{(\roman{enumi})}
\begin{enumerate}
\item $0 \leq I(d;V) \leq H(V)$ for all $V \subset A$.
\item For any $W \subset V \subset A$,
\begin{equation}
 0 \leq I(d;V)-I(d;W) \leq H(V)-H(W).
\end{equation}
\item For any $V,W \subset A$, 
\begin{multline*}
I(d;V)+I(d;W) - I(d;V \cup W) - I(d;V \cap W)\\ \leq H(V) + H(W) - H(V \cup W) - H(V \cap W).
\end{multline*}
\item $I(d;A) \leq y$.
\end{enumerate}
The marginal utility of information, $\mui(y)$ is defined as the derivative of $U(y)$.

We emphasize for clarity that, while we intuitively regard $I(d;V)$ as the information that a descriptor $d$ imparts about utility $V$, we formally treat the quantities $\{I(d;V) \}_{V \subset A}$ not as functions of two inputs but as variables subject to the above constraints.  

Throughout this appendix we consider a system $\cA=(A, H_\cA)$ comprising two independent subsystems, $\cB=(B, H_\cB)$ and $\cC=(C, H_\cC)$.  This means that $A$ is the disjoint union of $B$ and $C$, and $H(A)=H(B)+H(C)$.  The additivity property of MUI can be stated as 
\begin{equation}
\mui_\cA(y) = \min_{\substack{y_1 + y_2 = y\\y_1,y_2 \geq 0}} \max \big \{ \mui_\cB(y_1), \mui_\cC(y_2) \big \}.
\end{equation}

Alternatively, this property can be stated in terms of the reflection $\tilde{\mui}_\cA(x)$ of $\mui_\cA(y)$, with the dependent and independent variables interchanged (see Section \ref{sec:MUI}), as
\begin{equation}
\tilde{\mui}_\cA(x) = \tilde{\mui}_\cB(x) + \tilde{\mui}_\cC(x).
\end{equation}

The proof of this property is organized as follows.  Our first major goal is Theorem \ref{dhatadditive}, which asserts that $I(d;A)=I(d;B)+I(d;C)$ when $u$ is maximized.  Lemmas \ref{dsuperadditive} and \ref{dsubsetadditive} are technical relations needed to achieve this result.  We then apply the decomposition principle of linear programming to prove an additivity property of $U_\cA$ (Theorem \ref{OUIadditivecor}).  Theorem \ref{MUIadditivethm} then deduces the additivity of $\mui_\cA$ from the additivity of $\hat{U}_\cA$.  Finally, in Corollary \ref{cor:MUIreflection}, we demonstrate the additivity of the reflected function $\tilde{\mui}_\cA$.

\begin{lemma}
\label{dsuperadditive}
Suppose the quantities $\{I(d;V) \}_{V \subset A}$ satisfy Constraints (i)--(iv).  Then for any subset $V \subset A$,
\begin{equation}
I(d;V) \geq I(d;V \cap B) + I(d; V \cap C).
\end{equation}
\end{lemma}

\begin{proof}
Applying Constraint (iii) to the sets $V \cap B$ and $V \cap C$ we have
\begin{equation}
I(d;V \cap B)+I(d;V \cap C) - I(d;V) \\ \leq H(V \cap B) + H(V \cap C) - H(V).
\end{equation}
But by Lemma \ref{lem:twoindepsubsystems}, $H(V) = H(V \cap B) + H(V \cap C)$.  Thus the right-hand side above is zero, which proves the claim.
\end{proof}

\begin{lemma}
\label{dsubsetadditive} Suppose the quantities $\{I(d;V) \}_{V \subset A}$ satisfy Constraints (i)--(iv).
Suppose further that $W \subset V \subset A$ and $I(d;V) = I(d;V \cap B) + I(d; V \cap C).$  
Then $I(d;W) = I(d;W \cap B) + I(d; W \cap C)$.
\end{lemma}

\begin{proof}
Constraint (iii), applied to the sets $V \cap B$ and $W \cup (V \cap C)$, yields
\begin{multline}
\label{WcupVcapC}
I(d;V \cap B) + I \big(d; W \cup (V \cap C) \big) - I(d;V) - I(d;W \cap B)\\ 
\leq H(V \cap B) + H\big( W \cup (V \cap C) \big) - H(V) - H(W \cap B).
\end{multline}
By Lemma \ref{lem:twoindepsubsystems}, we have
\begin{align}
\label{hWcupVcapC}
H\big( W \cup (V \cap C) \big) & = H(W \cap B) + H(V \cap C)\\
\nonumber
H(V) & = H(V \cap B) + H(V \cap C).
\end{align}
With these two relations, the right-hand side of \eqref{WcupVcapC} simplifies to zero.  Making this simplification and substituting $I(d;V) = I(d;V \cap B) + I(d; V \cap C)$ (as given), we obtain
\begin{equation}
\label{WcupVcapC2}
I \big(d; W  \cup (V \cap C) \big) - I(d;W \cap B) - I(d; V \cap C) \leq 0.
\end{equation}

We next apply Constraint (iii) to $V \cap C$ and $W$, yielding
\begin{multline}
\label{WcupVcapC3}
I(d; V \cap C) + I (d; W) - I \big(d; W \cup (V \cap C) \big) - I(d; W \cap C)\\
\leq H(V \cap C) + H(W) - H(W \cup (V \cap C) \big) - H(W \cap C).
\end{multline}
Lemma \ref{lem:twoindepsubsystems} implies $H(W) = H(W \cap B) + H(W \cap C)$.  Combining this relation with \eqref{hWcupVcapC}, the right-hand side of \eqref{WcupVcapC3} simplifies to zero.  We then rewrite \eqref{WcupVcapC3} as
\begin{equation}
I (d; W) - I(d; W \cap C) \leq  I \big(d; W \cup (V \cap C) \big) -I(d; V \cap C).
\end{equation}
By \eqref{WcupVcapC2}, the right-hand side above is less than or equal to $I(d;W \cap B)$.  Making this substitution and rearranging, we obtain
\begin{equation}
I(d;W) \leq I(d; W \cap B) + I(d; W \cap C).
\end{equation}
Combining now with Lemma \ref{dsuperadditive}, it follows that $I(d;W) = I(d; W \cap B) + I(d; W \cap C)$ as desired.
\end{proof}

\begin{theorem}
\label{dhatadditive}
Suppose the quantities $\{I(\hat{d};V) \}_{V \subset A}$ maximixe $u=\sum_{a \in A} \sigma(a) I(d;a)$ subject to Constraints (i)--(iv) for some $0 \leq y \leq H(A)$. Then
\begin{equation}
I(\hat{d};A) = I(\hat{d};B) + I(\hat{d};C).
\end{equation}
\end{theorem}

\begin{proof}
Let $\hat{u} = \sum_{a \in A} \sigma(a) I(\hat{d};a)$ be the maximal value of $u$.  By the duality principle of linear programming, the quantities $\{I(\hat{d};V) \}_{V \subset A}$ minimize the value of $I(d;A)$ as $\{I(d;V) \}_{V \subset A}$ varies subject to Constraints (i)--(iii) along with the additional constraint $u \geq \hat{u}$.  (Informally, the descriptor $\hat{d}$ achieves utility $\hat{u}$ using minimal information.)

Assume for the sake of contradiction that $I(\hat{d};A) > I(\hat{d};B) + I(\hat{d};C)$.  We will obtain a contradiction by showing that there is another set of quantities $\{I(\tilde{d};V) \}_{V \subset A}$, satisfying (i)--(iii) and $\tilde{u} =\hat{u}$, with $I(\tilde{d};A)<I(\hat{d};A)$.  Here, $\tilde{u}$ is the utility associated to $\{I(\tilde{d};V) \}_{V \subset A}$; that is, $\tilde{u}=\sum_{a \in A} \sigma(a) I(\tilde{d};a)$.  (Informally, we construct a new descriptor $\tilde{d}$ that achieves the same utility as $\hat{d}$ using less information.)

To obtain such quantities $\{I(\tilde{d};V) \}_{V \subset A}$, we first define $S \subset 2^A$ as the set of all subsets $V \subset A$ that satisfy 
\begin{equation}
I(\hat{d};V) > I(\hat{d};V \cap B) + I(\hat{d}; V \cap C).
\end{equation}
We observe that, by Lemma \ref{dsuperadditive}, if $V \notin S$, then $I(\hat{d};V) = I(\hat{d};V \cap B) + I(\hat{d}; V \cap C)$.  It then follows from Lemma \ref{dsubsetadditive} that if $W \subset V \subset A$ and $W \in S$, then $V \in S$ as well.

Next we choose $\epsilon > 0$ sufficiently small that, for each $V \in S$, the following two conditions are satisfied:
\renewcommand{\labelenumi}{(\arabic{enumi})}
\begin{enumerate}
\item $I(\hat{d};V) > I(\hat{d};V \cap B) + I(\hat{d}; V \cap C) + \epsilon,$
\item $I(\hat{d};V) > I(\hat{d};W) + \epsilon$, for all $W \subset V, W \notin S$.
\end{enumerate}
There is no problem arranging for condition (2) to be satisfied for any particular $V \in S$, since it follows readily from Constraint (ii) on $\hat{d}$ that if $W \subset V$ and $W \notin S$, then $I(\hat{d};V) > I(\hat{d};W)$.  We also note that since $A$ is finite, there are only a finite number of conditions to be satisfied as $V$ and $W$ vary, so it is possible to choose an $\epsilon >0$ satisfying all of them.

Having chosen such an $\epsilon$, we define the quantities $\{I(\tilde{d};V) \}_{V \subset A}$ by
\begin{equation}
I(\tilde{d};V) = \begin{cases} I(\hat{d};V) - \epsilon &  V \in S\\
I(\hat{d};V) & \text{otherwise.} \end{cases}
\end{equation}
In words, we reduce the amount of information that is imparted about the sets in $S$ by an amount $\epsilon$, while leaving fixed the amount that is imparted about sets not in $S$.  Intuitively, one could say that we are exploiting an inefficiency in the amount of information imparted by $\hat{d}$ about sets in $S$, and that the new descriptor $\tilde{d}$ is more efficient in terms of minimizing $I(d;A)$ while maintaining $U(d) \geq \hat{u}$.

We will now show that $\tilde{d}$ satisfies Constraints (i)--(iii) and $U(\tilde{d}) =\hat{u}$.  First, since $0 \leq I(\tilde{d};V) \leq I(\hat{d};V) \leq H(V)$ for all $V \subset A$, Constraint (i) is clearly satisfied. 

For Constraint (ii), consider any $W \subset V \subset A$.  If $V$ and $W$ are either both in $S$ or both not in $S$ then $I(\tilde{d};V)-I(\tilde{d};W)=I(\hat{d};V)-I(\hat{d};W)$, and Constraint (ii) is satisfied for $\tilde{d}$ since it is satisfied for $\hat{d}$.  It only remains to consider the case that $V \in S$ and $W \notin S$.  In this case, we have
\begin{equation}
I(\tilde{d};V)-I(\tilde{d};W) = I(\hat{d};V) - I(\hat{d};W) - \epsilon > 0,
\end{equation}
since $V$ and $\epsilon$ satisfy condition (2) above.  Furthermore,
\begin{align*}
I(\tilde{d};V)-I(\tilde{d};W) & = I(\hat{d};V)-I(\hat{d};W) - \epsilon\\
& \leq H(V) - H(W) - \epsilon\\
& < H(V) - H(W).
\end{align*}
Thus Constraint (ii) is satisfied.

To verify Constraint (iii), we must consider a number of cases, only one of which is nontrivial.
\begin{itemize}
\item If either
\begin{itemize}
\item none of $V$, $W$, $V \cup W$ and $V \cap W$ belong to $S$,
\item all of $V$, $W$, $V \cup W$ and $V \cap W$ belong to $S$,
\item $V$ and $V \cup W$ belong to $S$ while $W$ and $V \cap W$ do not, or
\item $W$ and $V \cup W$ belong to $S$ while $V$ and $V \cap W$ do not,
\end{itemize}
then the difference on the left-hand side of Constraint (iii) has the same value for $d=\hat{d}$ and $d=\tilde{d}$---that is, the changes in each term cancel out in the difference.  Thus Constraint (iii) is satisfied for $\tilde{d}$ since it is satisfied for $\hat{d}$. 
\item If $V$, $W$, and $V \cup W$ belong to $S$ while $V \cap W$ does not, then
\begin{multline*}
I(\tilde{d};V)+I(\tilde{d};W) - I(\tilde{d};V \cup W) - I(\tilde{d};V \cap W)\\
= I(\hat{d};V)+I(\hat{d};W) - I(\hat{d};V \cup W) - I(\hat{d};V \cap W) - \epsilon.
\end{multline*}
The left-hand side of Constraint (iii) therefore decreases when moving from $d=\hat{d}$ to $d=\tilde{d}$.  So Constraint (iii) is satisfied for $\tilde{d}$ since it is satisfied for $\hat{d}$.  
\item The nontrivial case is that $V \cup W$ belongs to $S$ while $V$, $W$ and $V \cap W$ do not. Then
\begin{multline}
\label{nontrivialSS}
I(\tilde{d};V)+I(\tilde{d};W) - I(\tilde{d};V \cup W) - I(\tilde{d};V \cap W)\\
= I(\hat{d};V)+I(\hat{d};W) - \left( I(\hat{d};V \cup W) - \epsilon \right) - I(\hat{d};V \cap W).
\end{multline}
By the definition of $S$ and condition (1) on $\epsilon$, we have
\begin{align*}
I(\hat{d};V \cup W) - \epsilon & > I \left(\hat{d}; (V \cup W) \cap B \right) + I \left(\hat{d}; (V \cup W) \cap C \right) \\
I(\hat{d};V) & = I(\hat{d};V \cap B) + I(\hat{d};V \cap C)\\
I(\hat{d};W) & = I(\hat{d};W \cap B) + I(\hat{d};W \cap C)\\
I(\hat{d};V \cap W)  & = I \left(\hat{d}; (V \cap W) \cap B \right) + I \left(\hat{d}; (V \cap W) \cap C \right).
\end{align*}
Substituting into \eqref{nontrivialSS} we have
\begin{align*}
& I(\tilde{d};V)+I(\tilde{d};W) - I(\tilde{d};V \cup W) - I(\tilde{d};V \cap W)\\
& \quad < I(\hat{d};V \cap B)+I(\hat{d};W \cap B) \\
& \qquad  - I \big(\hat{d}; (V \cup W) \cap B \big) - I \big(\hat{d}; (V \cap W) \cap B \big) \\
& \qquad \quad + I(\hat{d};V \cap C)+I(\hat{d};W \cap C) \\
& \qquad \qquad- I \big(\hat{d}; (V \cup W) \cap C \big) - I \big(\hat{d}; (V \cap W) \cap C \big). 
\end{align*}
Applying Constraint (iii) on $\hat{d}$ twice to the right-hand side above, we have
\begin{multline*}
I(\tilde{d};V)+I(\tilde{d};W) - I(\tilde{d};V \cup W) - I(\tilde{d};V \cap W)\\
< H(V \cap B) + H(W \cap B) - H \big( (V \cup W) \cap B \big) - H \big( (V \cap W) \cap B \big)\\
\quad + H(V \cap C) + H(W \cap C) - H \big( (V \cup W) \cap C \big) - H \big( (V \cap W) \cap C \big).
\end{multline*}
But Lemma \ref{lem:twoindepsubsystems} implies that $H(Z \cap B) + H(Z \cap C)=H(Z)$ for any subset $Z \subset A$.  We apply this to the sets $V$, $W$, $V \cup W$ and $V \cap W$ to simplify the right-hand side above, yielding
\begin{multline*}
I(\tilde{d};V)+I(\tilde{d};W) - I(\tilde{d};V \cup W) - I(\tilde{d};V \cap W)\\
< H(V) + H(W) - H(V\cup W) - H(V \cap W),
\end{multline*}
as required.
\end{itemize}
No other cases are possible, since, as discussed above, any superset of a set in $S$ must also be in $S$.

Finally, it is clear that no singleton subsets of $A$ are in $S$.  Thus $I(\tilde{d};a) = I(\hat{d};a)$ for each $a \in A$, and it follows that $\sum_{a \in A} \sigma(a) I(\tilde{d};a) =\hat{u}$.  

We have now verified that $\tilde{d}$ satisfies Constraints (i)--(iii) and $U(\tilde{d})=\hat{u}$.  Furthermore, since $A \in S$ by assumption, we have $I(\tilde{d};A) < I(\hat{d};A)$.  This contradicts the assertion that $\hat{d}$ minimizes $I(d;A)$ subject to Constraints (i)--(iii) and $U(d) \geq \hat{u}$.  Therefore our assumption that $I(\hat{d};A)>I(\hat{d};B)+I(\hat{d};C)$ was incorrect, and we must instead have $I(\hat{d};A)=I(\hat{d};B)+I(\hat{d};C)$.
\end{proof}

\begin{theorem}
\label{OUIadditivecor}
The maximal utility function $U(y)$ is additive over independent subsystems in the sense that
\begin{equation}
U_\cA(y) 
= \max_{\substack{y_1 + y_2 = y\\y_1,y_2 \geq 0}} \left( U_\cB(y_1) + U_\cC(y_2) \right).
\end{equation}
\end{theorem}

\begin{proof}
For a given $y \geq 0$, let $\{I(\hat{d};V) \}_{V \subset A}$ maximixe $u=\sum_{a \in A} \sigma(a) I(d;a)$ subject to Constraints (i)--(iv). Combining Theorem \ref{dhatadditive} with Lemma \ref{dsubsetadditive}, it follows that $I(\hat{d}; V) = I(\hat{d}; V \cap B) + I(\hat{d}; V \cap C)$.  We may therefore augment our linear program with the additional constraint,
\renewcommand{\labelenumi}{(\roman{enumi})}
\begin{enumerate}
\setcounter{enumi}{4}
\item $I(d; V) = I(d; V \cap B) + I(d; V \cap C)$,
\end{enumerate}
for each $V \subset A$, without altering the optimal solution.

Upon doing so, we can use this new constraint to eliminate the variables $I(d;V)$ for $V$ not a subset of either $B$ or $C$.  We thereby reduce the set of variables from $\big \{ I(d;V) \big \}_{V \subset A}$ to
\begin{equation}
\big \{ I(d;V) \big \}_{V \subset B} \quad \cup \quad \big \{ I(d;W) \big \}_{W \subset C}.
\end{equation}

We observe that this reduced linear program has the following structure: The variables in the set $\big \{ I(d;V) \big \}_{V \subset B}$ are restricted by Constraints (i)--(iii) as applied to these variables.  Separately, variables in the set $\big \{ I(d;W) \big \}_{W \subset C}$ are also restricted by Constraints (i)--(iii), as they apply to the variables in this second set.  The only constraint that simultaneously involves variables in both sets is (iv).  This constraint can be rewritten as
\begin{equation}
u_\cB  + u_\cC \leq y,
\end{equation}
with
\begin{equation}
u_\cB = \sum_{b \in B} \sigma(b) I(d;b), \qquad u_\cC = \sum_{c \in C} \sigma(c) I(d;c).
\end{equation}

This structure enables us to apply the decomposition principle for linear programs \cite{dantzig1960decomposition} to decompose the full program into two linear sub-programs, one on the variables $\big \{ I(d;V) \big \}_{V \subset B}$ and one on $\big \{ I(d;W) \big \}_{W \subset C}$, together with a coordinating program described by Constraint (iv).  The desired result then follows from standard theorems of linear program decomposition \cite{dantzig1960decomposition}.  
\end{proof}

\begin{theorem} 
\label{MUIadditivethm}
$M_\cA$ is additive over independent subsystems in the sense that\begin{equation}
M_\cA(y) = \min_{\substack{y_1 + y_2 = y\\y_1,y_2 \geq 0}} \max \big \{ M_\cB(y_1), M_\cC(y_2) \big \}.
\end{equation}
\end{theorem}

\begin{proof}
We define the function
\begin{equation}
F(y_1;y) = U_\cB(y_1) + U_\cC(y-y_1).
\end{equation}
The result of Theorem \ref{OUIadditivecor} can then be expressed as 
\begin{equation}
U(y) = \max_{0 \leq y_1 \leq y} F (y_1; y).
\end{equation}
We choose and fix an arbitrary $y$-value $\tilde{y} \geq 0$, and we will prove the desired result for $y=\tilde{y}$. 

We observe that $F(y_1; \tilde{y})$ is concave in $y_1$ since $U_\cB(y_1)$ and $U_\cC(\tilde{y}-y_1)$ are.  It follows that any local maximum of $F(y_1;\tilde{y})$ in $y_1$ is also a global maximum.   We assume that the maximum of $F(y_1; \tilde{y})$ in $y_1$ is achieved at a single point $\hat{y}_1$ with $0 < \hat{y}_1 < \tilde{y}$.  The remaining cases---that the maximum is achieved at $y_1=0$ or $y_1=\tilde{y}$, or is achieved on a closed interval of $y_1$-values---are trivial extensions of this case.

Assuming we are in the case described above (and again invoking the concavity of $F$ in $y_1$), $\hat{y}_1$ must be the unique point at which the derivative $\frac{\partial F}{\partial y_1} (y_1; \tilde{y})$ changes sign from positive to negative.  This derivative can be written
\begin{equation}
\frac{\partial F}{\partial y_1} (y_1; \tilde{y}) = M_\cB(y_1) - M_\cC(\tilde{y}-y_1).
\end{equation}
It follows that $\hat{y}_1$ is the unique real number in $[0,\tilde{y}]$ satisfying
\begin{equation}
\label{yhatinequalities}
\begin{cases}
M_\cB(y_1) > M_\cC(\tilde{y}-y_1) & y_1 < \hat{y}_1\\
M_\cB(y_1) < M_\cC(\tilde{y}-y_1) & y_1 > \hat{y}_1.
\end{cases}
\end{equation}

From inequalities \eqref{yhatinequalities}, and using the fact that $M_\cB(y_1)$ and $M_\cC(y_2)$ are nonincreasing, piecewise-constant functions, we see that either $M_\cB(y_1)$ decreases at $y_1=\hat{y}_1$, or $M_\cC(y_2)$ decreases at $y_2=\tilde{y}-\hat{y}_1$, or both.  We analyze these cases separately.

\paragraph*{Case 1: $M_\cB(y_1)$ decreases at $y_1=\hat{y}_1$, while $M_\cC(y_2)$ is constant in a neighborhood of $y_2=\tilde{y}-\hat{y}_1$.} Pick $\epsilon>0$ sufficiently small so that $M_\cC(y_2)$ has constant value for $y_2 \in (\tilde{y}-\hat{y}_1-\epsilon, \tilde{y}-\hat{y}_1+\epsilon)$.  Then inequalities \eqref{yhatinequalities} remain satisfied with $\tilde{y}$ replaced by any $y \in (\tilde{y}-\epsilon, \tilde{y}+\epsilon)$ and $\hat{y}_1$ fixed.  Thus for $y$ in this range, we have
\begin{equation}
U_\cA(y) = U_\cB(\hat{y}_1) + U_\cC(y-\hat{y}_1).
\end{equation}
Taking the derivative of both sides in $y$ at $y=\tilde{y}$ yields
\begin{equation}
\label{MAytilde}
M_\cA(\tilde{y})=M_\cC(\tilde{y}-\hat{y}_1).
\end{equation}
We claim that 
\begin{equation}
\label{Mminmaxclaim}
M_\cC(\tilde{y}-\hat{y}_1) 
=  \min_{0 \leq y_1 \leq \tilde{y}} \max \big \{ M_\cB(y_1), M_\cC(\tilde{y} - y_1) \big \}.
\end{equation}
To prove this claim, we first note that, by the inequalities \eqref{yhatinequalities},
\begin{equation}
\label{eq:Mcases}
\max \big \{ M_\cB(y_1), M_\cC(\tilde{y} - y_1) \big \}=
\begin{cases}
M_\cB(y_1)  & y_1 < \hat{y}_1\\
M_\cC(\tilde{y}-y_1) & y_1 > \hat{y}_1.
\end{cases}
\end{equation}
Since both $M_\cB$ and $M_\cC$ are piecewise-constant and nonincreasing, the minimax in Eq.~\eqref{Mminmaxclaim} is achieved for values $y_1$ near $\hat{y}_1$.  We therefore can restrict to the range $y_1 \in (\hat{y}_1 - \epsilon, \hat{y}_1 + \epsilon)$.  Combining Eq.~\eqref{eq:Mcases} with the conditions defining Case 1 and the definition of $\epsilon$, we have 
\begin{equation}
\begin{split}
\max \big \{ M_\cB(y_1), M_\cC(\tilde{y} - y_1) \big \} = M_\cB(y_1) > M_\cC(\tilde{y} - \hat{y}_1) &
 \qquad \text{for $y_1 \in (\hat{y}_1 - \epsilon, \hat{y}_1)$}\\
\max \big \{ M_\cB(y_1), M_\cC(\tilde{y} - y_1) \big \} = M_\cC(\tilde{y} - y_1) = M_\cC(\tilde{y} - \hat{y}_1) &
 \qquad \text{for $y_1 \in (\hat{y}_1, \hat{y}_1 + \epsilon)$.}
\end{split}
\end{equation}
Thus the minimax in Eq.~\eqref{Mminmaxclaim} is achieved at a value of $M_\cC(\tilde{y} - \hat{y}_1)$ when $y_1 \in (\hat{y}_1, \hat{y}_1 + \epsilon)$, verifying Eq.~\eqref{Mminmaxclaim}.  Combining with Eq.~\eqref{MAytilde}, we have
\begin{equation}
M_\cA(y) = \min_{\substack{y_1 + y_2 = y\\y_1,y_2 \geq 0}} \max \big \{ M_\cB(y_1), M_\cC(y_2) \big \},
\end{equation}
proving the theorem in this case.

\paragraph*{Case 2: $M_\cB(y_1)$ is constant in a neighborhood of $y_1=\hat{y}_1$, while $M_\cC(y_2)$ decreases at $y_2 =\tilde{y}-\hat{y}_1$.}  In this case, we define $\hat{y}_2 = \tilde{y}-\hat{y}_1$.  The proof then follows exactly as in Case 1, with $\cB$ and $\cC$ interchanged, and the subscripts 1 and 2 interchanged.

\paragraph*{Case 3: $M_\cB(y_1)$ decreases at $y_1=\hat{y}_1$ and $M_\cC(y_2)$ decreases at $y_2 =\tilde{y}-\hat{y}_1$.}  This case only occurs at the $y$-values for which $U_\cA(y)$ changes slope and $M_\cA(y)$ changes value.  At these nongeneric points, $M_\cA(y)$ (defined as the derivative of $U_\cA(y)$) is undefined.  We therefore disregard this case. 
\end{proof}

We now define $\tilde{M}_\cA(x)$ as the reflection of $M_\cA(y)$ with the dependent and independent variables interchanged.  Since $M_\cA$ is positive and nonincreasing, $\tilde{M}_\cA$ is a well-defined function given by the formula
\begin{equation}
\tilde{\mui}_\cA(x)  = \max\{y:\mui_\cA(y) \leq x \}.
\end{equation}
The following corollary gives a simpler expression of the additivity property of MUI.  

\begin{corollary}
\label{cor:MUIreflection}
If $\cA$ consists of independent subsystems $\cB$ and $\cC$, then $\tilde{\mui}_\cA(x) = \tilde{\mui}_\cB(x) + \tilde{\mui}_\cC(x)$ for all $x \geq 0$.
\end{corollary}

\begin{proof}
Combining the above formula for $\tilde{\mui}_\cA(x)$ with the result of Theorem \ref{MUIadditivethm}, we write
\begin{align*}
\tilde{\mui}_\cA(x)  & = \max\{y:\mui_\cA(y) \leq x \}\\
& = \max \left \{y: \left( \min_{\substack{y_1 + y_2 = y\\y_1,y_2 \geq 0}} \max \big \{ M_\cB(y_1), M_\cC(y_2) \big  \} \right) \leq x \right \}\\
& = \max \left \{y: \Big (\exists y_1, y_2 \geq 0 : 
\left( y_1 + y_2 = y \text{ AND } \max \big \{ M_\cB(y_1), M_\cC(y_2) \big \} \leq x \right) \Big) \right \} \\
& = \max \left \{(y_1 + y_2) : 
\Big (y_1, y_2 \geq 0 \text{ AND } M_\cB(y_1) \leq x \text{ AND } M_\cC(y_2) \leq x \Big) \right \} \\
& = \max \{ y_1 : M_\cB(y_1) \leq x \} + \max \{ y_2 : M_\cB(y_2) \leq x \} \\
& = \tilde{\mui}_\cB(x) + \tilde{\mui}_\cC(x). \qedhere
\end{align*}
\end{proof}

\section{Marginal Utility of Information for Parity Bit Systems}
\label{app:MUIpbit}

Here we compute the MUI for a family of systems which exhibit exchange
symmetry and have a constraint at the largest scale.  Systems in this
class have $N \geq 3$ components and information function given by
\begin{equation}
H(V) = H_{|V|} = \begin{cases}
 |V| & |V| \leq N-1 \\
 N-1 & |V| = N.
\end{cases}
\label{eq:paritybitH}
\end{equation}
This includes example \exD{} as the case $N = 3$.  More generally,
this family includes systems of $N-1$ independent random bits together
with one parity bit.

Since these systems have exchange symmetry, we expect an optimal description to have exchange symmetry as well; that is $I(d;U)$ should depend only on the number of components in $U$.  We therefore use the simplified notation $I_n$ for information that $d$ imparts about any set of $n$ components, $0 \leq n \leq N$.

We begin by establishing three relations among the $I_n$.  First, Constraint~(ii) in Section \ref{sec:MUI} tells us that
\begin{equation}
0 \leq I_N- I_{N-1}
 \leq H_N - H_{N-1}.
\end{equation}
But the right side of this expression vanishes ($H_N = H_{N-1}=N-1$), so we have that 
\begin{equation}
\label{eq:INminus1}
I_{N-1}=I_N.
\end{equation}  

Second, Constraint (iii) applied to disjoint subsets of size $m$ and $n$, with $m+n \leq N-1$, implies that
\begin{equation}
I_m + I_n - I_{m+n}
 \leq H_m + H_n - H_{m+n}.
\end{equation}
The right-hand side of the inequality vanishes by Eq.~\eqref{eq:paritybitH}, so 
\begin{equation}
\label{eq:pbitsubadditive}
I_m + I_n \leq I_{m+n} \qquad \text{for all $m,n \geq 0, \; m+n \leq N-1$.}
\end{equation}
By iteratively applying Eq.~\eqref{eq:pbitsubadditive} we arrive at the inequality
\begin{equation}
\label{eq:I1}
(N-1) I_1 \leq I_{N-1}.
\end{equation}

Third, Constraint (iv) implies that 
\begin{equation}
\label{eq:IN}
I_N \leq y.
\end{equation}

Combining \eqref{eq:INminus1},  \eqref{eq:I1} and \eqref{eq:IN}, we have
\begin{equation}
\label{eq:I1y}
I_1 \leq \frac{I_{N-1}}{N-1} = \frac{I_N}{N-1} \leq \frac{y}{N-1}.
\end{equation}
By definition, the utility of a descriptor in this system satisfies
\begin{equation}
\label{eq:uI1}
u = N I_1.
\end{equation}
Combining with \eqref{eq:I1y} yields the inequality
\begin{equation}
\label{eq:uy}
u \leq \frac{N}{N-1} y.
\end{equation}

Inequality \eqref{eq:uy} places a limit on the utility of any descriptor satisfying exchange symmetry.  To complete the argument, we exhibit a descriptor for which equality holds in~\eqref{eq:uy}.  This descriptor is defined by 
\begin{equation}
I_m = \begin{cases}
 \frac{m}{N-1} \min\{y,N-1\} & 0 \leq m \leq N-1 \\
  \min\{y,N-1\} & m = N.
\end{cases}
\label{eq:paritybitI}
\end{equation}
It is straightforward to verify that Constraints (i)--(iv) are satisfied by this descriptor.  Combining Eqs.~\eqref{eq:uI1} and \eqref{eq:paritybitI}, we have that for $0 \leq y \leq N-1$,
\begin{equation}
u = \frac{N}{N-1} y.
\end{equation}
By inequality \eqref{eq:uy}, this descriptor achieves optimal utility.  
We therefore have
\begin{equation}
\boxed{U(y) = \frac{N}{N-1}y,\qquad M(y) = \frac{N}{N-1},
\qquad 0 \leq y \leq N - 1.}
\end{equation}
Setting $N = 3$, we recover the MUI for example \exD{} as stated in
the main text, Eq.~(\ref{eq:MUI-exD-maintext}).

\section{Complexity Profile and Kinetic Theory}
\label{app:green}
The complexity profile, Eq.~(\ref{eq:CP}), presented in
\cite{baryam2004b}, has an intellectual antecedent in a series
expansion for entropy introduced by Green in kinetic theory
\cite{green1952,nettleton1958}.  Though this familial relationship has
been acknowledged in the literature \cite{sgs2004}, it has yet to be
studied in detail.  We shall do so in this appendix.

In kinetic theory, we deal with large numbers of particles, at least
comparable to Avogadro's number ($\approx 10^{23}$).  We use
statistical methods to deduce macroscopic properties of the aggregate
from our knowledge of the microscopic interactions among individual
atoms~\cite{kardar2007}.  A microstate of the system is uniquely
identified by specifying the positions and momenta of all the atoms.
If we are uncertain about what the system's microstate might be, we
can encapsulate our knowledge of the system in a probability
distribution defined over the set of all possible microstates, which
we write
\begin{equation}
\rho = \rho(\vec{p}_1,\vec{p}_2,\ldots,\vec{p}_N,
            \vec{q}_1,\vec{q}_2,\ldots,\vec{q}_N).
\end{equation}

Often, we do not care \emph{which} particle is doing something, only
that \emph{any} particle is.  So, we define a one-particle
distribution function by projecting $\rho$ down to a single particle.
By further assuming that the density $\rho$ is symmetric under
particle exchange, we can write the one-particle distribution function as
\begin{equation}
f_{\rm I}(\vec{p},\vec{q},t) = N \int \prod_{i = 2}^N d^3\vec{p}_i d^3\vec{q}_i
   \,\rho(\vec{p}_1 = \vec{p}, \vec{q}_1 = \vec{q},
          \vec{p}_2, \vec{q}_2, \ldots, \vec{p}_N, \vec{q}_N, t).
\end{equation}
A two-body probability density can be defined in a similar way:
\begin{equation}
f_{\rm II}(\vec{p}_1,\vec{q}_1,\vec{p}_2,\vec{q}_2,t) =
   N(N - 1) \int \prod_{i = 3}^N dV_i\, \rho(\vec{p}_1,\vec{q}_1,
                                           \vec{p}_2,\vec{q}_2,\ldots,
                                           \vec{p}_N,\vec{q}_N,t).
\label{eq:fII}
\end{equation}

Green \cite{green1952} provides a way of estimating the Shannon
entropy of the full phase-space distribution $\rho$ in terms of
lower-scale correlation functions $f_l$.  The series expansion for
$S[\rho]$ involves particular logarithmic transforms of the functions
$f_l$, which we will now investigate.

Green suggests looking at the quantity $z_{\rm II}^{ij}$, defined by
the following relation:
\begin{equation}
\exp z_{\rm II}^{ij} = \frac{f_{\rm II}^{ij}}{f_{\rm I}^i f_{\rm
    I}^j}.
\label{eq:zII}
\end{equation}
This quantity indicates the probability of finding two molecules in a
given configuration, {\em relative} to the case where they do not
influence each other.  When $z_{\rm II}^{ij}$ is zero, then its
exponential is unity, and the molecules are statistically independent.

If we define the first-scale $z$-function by
\begin{equation}
\log f_{\rm I}^i = z_{\rm I}^i,
\end{equation}
then we can rewrite Eq.~(\ref{eq:zII}) in the following suggestive
way:
\begin{equation}
\log f_{\rm II}^{ij} = z_{\rm II}^{ij} + z_{\rm I}^{i} + z_{\rm I}^{j}.
\end{equation}
The logarithm of the second-scale $f$-function is a sum over
$z$-functions of first and second scale, and at each scale, all
possible groupings of that size are represented.  This suggests a
natural way to define $z_{\rm III}$ in terms of the three-body
function $f_{\rm III}$:
\begin{equation}
\log f_{\rm III}^{ijk} = z_{\rm III}^{ijk} 
   + z_{\rm II}^{ij} + z_{\rm II}^{jk} + z_{\rm II}^{ki}
   + z_{\rm I}^{i} + z_{\rm I}^{j} + z_{\rm I}^{k}.
\label{eq:fIII}
\end{equation}
%\begin{equation}
%\exp z_{\rm III}^{ijk} = f_{\rm III}^{ijk}
%   \frac{f_{\rm I}^i f_{\rm I}^j f_{\rm I}^k}
%        {f_{\rm II}^{ij} f_{\rm II}^{jk} f_{\rm II}^{ki}}.
%\end{equation}
Generally speaking, higher-scale $z_l$ are defined by writing the
logarithm of $f_l$ as the sum of $z$-functions of all scales up to
$l$, with each possible subset of $l$ molecules represented by a term.
Inverting these relations gives
\begin{align}
z_{\rm I}^i &= \log f_{\rm I}^i\\
z_{\rm II}^{ij} &= \log f_{\rm II}^{ij}
                  - \log f_{\rm I}^i - \log f_{\rm I}^j  \\
z_{\rm III}^{ijk} &= \log f_{\rm III}^{ijk}
                  - \log f_{\rm II}^{ij} - \log f_{\rm II}^{jk}
                  - \log f_{\rm II}^{ki}
                  + \log f_{\rm I}^i
                  + \log f_{\rm I}^j
                  + \log f_{\rm I}^k \label{eq:zIII}
\end{align}
Again, all possible groupings appear on the right-hand side, although
this time the signs are varied.  Each term has a prefactor $(-1)^l$,
where $l$ is the number of molecules ``left out'' from the group.  For
example, the coefficient on the $f_{\rm I}$ terms in $z_{\rm III}$ is
$(-1)^{{\rm III} - {\rm I}} = 1$.

The ``fine-grained'' entropy we wish to estimate is
\begin{equation}
S = -\frac{k_B}{N!} \int d\vec{q}_1\cdots d\vec{q}_N
                         d\vec{p}_1\cdots d\vec{p}_N\,
                         f_N \log f_N.
\end{equation}
Green observed that if all $N$ molecules are identical, the $\log f_N$
can be expanded in the following way:
\begin{equation}
S = -\frac{k_B}{N!} \int d\vec{q}_1\cdots d\vec{q}_N
                         d\vec{p}_1\cdots d\vec{p}_N\,
                         f_N \left[\binom{N}{1} z_{\rm I}^1
                                  + \binom{N}{2} z_{\rm II}^{12}
                                  + \ldots + z_N\right].
\end{equation}
Bringing the $1/N!$ into the integral cancels some of the
combinatorial factors, and others can be absorbed into the symmetry
factors which occur in the multi-particle distribution functions, as
in Eq.~(\ref{eq:fII}).  This yields
\begin{equation}
\boxed{S = -k_B \left[\frac{1}{1!}
                      \int d\vec{q}_1d\vec{p}_1\, f_{\rm I} z_{\rm I}
                      + \frac{1}{2!}
                        \int d\vec{q}_1d\vec{p}_1
                             d\vec{q}_2d\vec{p}_2\, f_{\rm II} z_{\rm II} +
               \ldots\right].}
\label{eq:green-expansion}
\end{equation}
Eq.~(\ref{eq:green-expansion}) is Green's expansion for the entropy.
It is a sum over scales:  the total entropy is given by a one-particle
contribution, plus a correction due to two-particle correlations, and
so on.

Note that the two-molecule correction term in
Eq.~(\ref{eq:green-expansion}) is, up to a symmetry factor, just the
mutual information between molecules, since
\begin{equation}
\int f_{\rm II} z_{\rm II} = \int f_{\rm II}^{12} 
    \log \frac{f_{\rm II}^{12}}{f_{\rm I}^1 f_{\rm I}^2}.
\end{equation}
When the molecules are uncorrelated, the mutual information vanishes
and the approximation gives the exact entropy, as expected.
Furthermore, the next correction, involving an integral over $f_{\rm
  III} z_{\rm III}$, is a constant factor times the {\em multivariate}
mutual information defined in Section~\ref{Dependencies}.  Compare the
signs in Eq.~(\ref{eq:C3}), where we wrote $C(3)$ for a
three-component system, to those in Eq.~(\ref{eq:zIII}), where we
defined $z_{\rm III}$: terms involving one or three components
(molecules) get a $+$ sign, while those involving two components
(molecules) get a $-$ sign.  The sign choices in the complexity
profile, Eq.~(\ref{eq:CP}), are revealed as the signs produced by
inverting the system of linear equations which define $f_l$ in terms
of $z_l$, as in Eq.~(\ref{eq:fIII}).

Wolf \cite{wolf1996} writes the total system entropy as a sum of
``information correlations'' equivalent to
Eq.~(\ref{eq:green-expansion}) and derives formulas which we can
identify as multivariate mutual information functions; however, to our
knowledge, the connection to $D(k)$ and $C(k)$ has not been made
explicit in the literature until now.  (To illustrate how mathematical
discovery happens: The thesis of Wolf \cite{wolf1996} rediscovers
multiple mutual information without naming it as such or connecting to
the literature \cite{mcgill1980,han1980,yeung1991,jakulin2003}.  Also,
it reinvents the composition of Joyal's {\em esp\`eces de structure,}
or combinatorial species \cite{joyal1981, bergeron1998, baez2000,
  morton2006}, without drawing the connection to combinatorial species
theory.)

Green's approximation has also seen some use in certain
\emph{nonequilibrium} molecular dynamics work, where one must consider
the \emph{time evolution} of the joint probability distribution
$\rho$.  In this context, the full sum over~$I_k$ is ill-behaved (due
to global constraints affecting the largest-$k$ contributions) and not
necessarily physically meaningful.  Truncating the sum at a small
value of~$k$ gives a more meaningful result~\cite{evans1989,
  evans2002}.  The issue of how a probability distribution like~$\rho$
depends on time is a subtle one.  When talking of dynamical systems,
our probability assignments really carry \emph{two} time indices: one
for the time when we have information in hand, and the other for the
time to which that information pertains.  If we ascribe probabilities
in a certain way about what the microstate might be at time $t$, and
we specify the interactions which can exist between individual atoms,
then our hands are forced: to be consistent with how we assign
probabilities for the microstate at $t$, we must make probability
assignments about what will happen at another time $t'$ in accord with
the Liouville equation~\cite{kardar2007}.  This is different than what
would happen if we gained new information at a later time and changed
our probability distribution accordingly.  

Having placed the Green expansion, Eq.~(\ref{eq:green-expansion}), in
its proper context, we can find applications for it beyond kinetic
theory.  For example, in the study of complex networks, one statistic
of interest is the Shannon entropy of a network's degree distribution.
This quantity is invoked when exploring, \emph{e.g.,} the response of
a system modeled by a network to an external attack or perturbation
\cite{wang2005}.  However, focusing on the degrees of a network's
nodes ignores the possibility of \emph{degree-degree correlations,}
which are known to be nontrivial in many cases of interest: the
probability that a node of degree $d$ is linked with another node of
degree $d'$ is not always computable knowing only the probability
distribution of node degrees $p(d)$ \cite{dorogovtsev2010, boguna2002,
  dorogovtsev2004}.  In turn, measuring the degree-degree correlation
itself fails to capture possible structure of higher rank
\cite{mahadevan2006}.  If we define a family of distribution functions
$f_l(d_1,\ldots,d_l)$, each giving the probability that the nodes in a
subnetwork of size $l$ have the degrees $d_1,\ldots,d_l$, then we can
use Green's approach to calculate the multivariate mutual information
content at scale $l$.  The overall degree-based complexity of the
network is then $C(1)$, as found by Eq.~(\ref{eq:green-joint-info}).
For a network which has no structure at rank $k \geq 2$, such as a
simple Erd\H{o}s--R\'enyi model, $C(1) = D(1)$.

\bibliographystyle{utphys}
\bibliography{multiscale.bib}

\end{document}